\documentclass[runningheads]{llncs}
\usepackage[T1]{fontenc}
\usepackage{amssymb, bussproofs, amsmath, datetime}
\usepackage{color}
\usepackage{proof}
\usepackage[lutzsyntax]{virginialake}
\usepackage{tikz}
\usepackage{listings}
\usepackage{enumerate}
\usepackage{bm}
\usepackage{hyperref}

\definecolor{grey}{rgb}{0.9, 0.9, 0.9}

\renewcommand{\qed}{\hfill$\blacksquare$}

\newcommand{\der}{derivation}

\newcommand{\prf}{proof}
\newcommand{\prfs}{proofs}

\newcommand{\prov}{provable}

\newcommand{\sys}[1]{\mathsf{#1}}
\newcommand{\Ra}{\Rightarrow}
\DeclareMathOperator{\slex}{<\!\!<}
\newcommand{\imp}{\rightarrow}
\renewcommand{\rule}[1]{(\text{#1})}

\newcommand{\goodbox}{\hspace{-.6ex}\text{
  \tikz[baseline=-.6ex, rounded corners=.01ex, line width=.1ex]
    {\draw (-.6ex,-.6ex) rectangle (.6ex,.6ex);}}\kern.2ex}
\newcommand{\gooddiamond}{\hspace{-.6ex}\text{
  \tikz[baseline=-.6ex, rounded corners=.01ex, rotate=45, line width=.1ex]
    {\draw (-.5ex,-.5ex) rectangle (.5ex,.5ex);}}\kern.2ex}
\renewcommand{\Box}{\goodbox}
\renewcommand{\Diamond}{\gooddiamond}

\begin{document}
\title{A new calculus for intuitionistic Strong L\"ob logic:
strong termination and cut-elimination, formalised}
\titlerunning{A new calculus for intuitionistic Strong L\"ob logic}
%
\author{Ian Shillito\inst{1} \and
Iris van der Giessen\inst{2} \and
Rajeev Gor{\'{e}}\inst{3,4} \and
Rosalie Iemhoff \inst{5}
}
\authorrunning{I. Shillito et al.}
%
\institute{Australian National University, Ngunnawal Country, Australia \and
University of Birmingham, United Kingdom \and
Technical University of Vienna, Austria \and
Polish Academy of Science, Poland \and
Utrecht University, The Netherlands}
\maketitle 
\begin{abstract}
We provide a new sequent calculus that enjoys syntactic cut-elimination and strongly terminating backward proof search for the intuitionistic Strong L\"ob logic $\sys{iSL}$, an intuitionistic modal logic with a provability interpretation.
A novel measure on sequents is used to prove both the termination of the naive backward proof search strategy, and the admissibility of cut in a syntactic and direct way, leading to a straightforward cut-elimination procedure.
All proofs have been formalised in the interactive theorem prover Coq.

\keywords{Intuitionistic provability logic, Cut-elimination, Backward proof search, Interactive theorem proving, Proof theory.}
\end{abstract}


\section{Introduction}

G\"odel-L\"ob logic~$\sys{GL}$ extends classical modal logic $\sys{K}$ with the G\"odel-L\"ob axiom
$\Box (\Box \varphi \imp \varphi) \imp \Box \varphi$. $\sys{GL}$ is the provability logic of Peano
Arithmetic $\sys{PA}$, i.e.~it consists of all modal formulas that are true under any arithmetical interpretation where
$\Box\varphi$ means ``$\varphi$ is provable in $\sys{PA}$'' (expressed in the language of~$\sys{PA}$).

An intuitionistic version of $\sys{GL}$ is $\sys{iGL}$ and the intuitionistic counterpart of~$\sys{PA}$ is Heyting Arithmetic
$\sys{HA}$. For a long time, the provability logic of $\sys{HA}$ was an open problem and was only known to be an extension of~$\sys{iGL}$.
However, Mojtahedi claims to have found a solution in a preprint~\cite{Mojtahedi22} currently under review.

Several other logics also have provability interpretations, such as modalised Heyting calculus~$\sys{mHC}$, Kuznetsov-Muravitsky logic~$\sys{KM}$, and
intuitionistic Strong L\"ob logic~$\sys{iSL}$ \cite{Esakia06,KuznetsovMuravitsky86,Muravitsky14,Lit14}. All these intuitionistic modal logics except $\sys{mHC}$ include the G\"odel-L\"ob axiom and all except $\sys{iGL}$ contain the so-called completeness axiom $\varphi \to \Box \varphi$.

Important to note is that these logics are defined over the language with only the $\Box$-modality and without $\Diamond$. In classical modal logic, $\Diamond$ is dual to $\Box$ and reads as consistency in the provability interpretation. However, for intuitionistic modal logics, in general, $\Diamond$ and $\Box$ are not interdefinable and several choices can be made. Interestingly, intuitionistic modal logics defined over the language with only the $\Box$ already reveal intrinsic intuitionistic characters. Important for us is the aforementioned completeness principle, also known as the coreflection principle. It trivializes in a classical setting, but has interesting intuitionistic readings. 
Indeed, in our setting of provability, $\varphi \rightarrow \Box \varphi$ reads as completeness: ``if~$\varphi$ is true then $\varphi$ is provable'' (see \cite{Visser82} for a discussion on the completeness principle in extensions of Heyting Arithmetic). The coreflection principle also appears in intuitionistic epistemic logic and lax logic (for overviews see, e.g., \cite{Lit14,Gie22}).

Here, we consider $\sys{iSL}$, the minimal intuitionistic modal logic with both the G\"odel-L\"ob axiom and the completeness axiom, which can also 
be axiomatised over intuitionistic modal logic $\sys{iK}$ by the Strong L\"ob axiom
$(\Box\varphi \to \varphi) \to \varphi$. The logic $\sys{iSL}$ is the provability logic of an extension of Heyting Arithmetic with
respect to so-called slow provability~\cite{VisserZoethout18} and plays an important role in the $\Sigma_1$-provability logic
of $\sys{HA}$~\cite{ArdMoj18}.

The G\"odel-L\"ob axiom characterises transitive converse well-founded Kripke frames for $\sys{GL}$ and also for
the birelational frames for $\sys{iGL}$, $\sys{iSL}$, and $\sys{KM}$. Interestingly, for $\sys{iSL}$, $\sys{mHC}$, and $\sys{KM}$,
the modal relation is a part of the intuitionistic relation.
This semantics plays an important role in the study of $\sys{iSL}$, e.g.~in the characterisation of its admissible rules \cite{Gie23}. A natural deduction system for $\sys{iSL}$ can be found in \cite{Bro22}. The proof systems that we focus on here are sequent calculi.

From a proof-theoretic perspective, the ``diagonal formula'' $\Box \varphi$ in the modal \rule{GLR} rule for $\sys{GL}$ 
causes difficulties for direct cut-elimination because the standard induction on the
size of the cut-formula and the height fail. Cut-elimination is highly
nontrivial as witnessed by decades of unsuccessful attempts and controversies before the proof by Valentini \cite{Valentini83} was
finally shown to be correct \cite{GoreRamanayake12}.

\begin{center}
\begin{math}
\begin{array}{c@{\extracolsep{1cm}}c}
    \AxiomC{$\Gamma, \Box \Gamma, \Box \varphi \Rightarrow \varphi$}
    \RightLabel{\rule{GLR}}
    \UnaryInfC{$\Phi, \Box \Gamma \Rightarrow \Box \varphi, \Delta$}
    \DisplayProof
  &
    \AxiomC{$\Gamma, \varphi \rightarrow \psi \Rightarrow \varphi$}
    \AxiomC{$\Gamma, \psi \Ra \varphi$}
    \RightLabel{\rule{$\rightarrow\!\text{L}_{\text{i}}$}}
    \BinaryInfC{$\Gamma, \varphi\rightarrow\psi \Ra \varphi$}
    \DisplayProof
\end{array}
\end{math}
\end{center}

In backward proof search, the \rule{GLR} rule causes loops because $\Box\Gamma$ is preserved upwards from conclusion to premise.
For \rule{GLR}, a simple terminating and complete strategy consists in applying \rule{GLR} only if $\Box\varphi\not\in\Box\Gamma$.
In sequent calculi for intuitionistic logic, the traditional \rule{$\rightarrow\!\text{L}_{\text{i}}$} rule, shown above right, can cause backward proof search to go into loops. 
For termination without loop check, various authors have independently discovered the sequent calculus $\sys{G4ip}$ which replaces the \rule{$\rightarrow\!\text{L}_{\text{i}}$} rule with multiple rules, depending on the form of $\varphi$~\cite{Dyc16}. Iemhoff~\cite{Iemhoff22} developed $\sys{G4}$-like calculi for several intuitionistic modal logics.

Thus, in a sequent calculus for an intuitionistic provability logic, both the modal rule and left implication rule have the potential
to cause loops \emph{and} the modal rule can complicate direct cut-elimination! For logic $\sys{iGL}$, van der Giessen and Iemhoff have developed $\sys{G3iGL}$ and $\sys{G4iGL}$~\cite{GieIem21}, providing a direct cut-elimination procedure for the former. The initial proof of cut-elimination for $\sys{G4iGL}$ was indirect, via $\sys{G3iGL}$, but Gor\'e
and Shillito later formalised direct cut-elimination using the maximal height of derivations as induction
parameter~\cite{GorShi22}.

Recently, van der Giessen and Iemhoff~\cite{GieIem23} developed two sequent calculi, $\sys{G3iSL}$ and $\sys{G4iSL}$, for
$\sys{iSL}$ for which they provided the analogue results compared to $\sys{G3iGL}$ and $\sys{G4iGL}$ mentioned above. In particular, they show that backward proof search in $\sys{G4iSL}$ \emph{weakly}
  terminates: \emph{there exists} a terminating (and complete) backward proof search strategy, namely one similar to the above-described for logic~$\sys{GL}$. However, \emph{not all} strategies terminate on this calculus: the naive backward proof search strategy, apply any rule in any order, does not.

Here, we present $\sys{G4iSLt}$ which replaces the $\sys{G4iSL}$ rules of the top row below, by the rules in the bottom row.
As suggested by van der Giessen and Iemhoff~\cite{GieIem23}, the new modal rule drops the explicit embedding of transitivity. But crucially, the new left-implication rule drops both transitivity and contraction on $\Box\varphi\imp\psi$ in the left premise. The right premise $S=\Phi,\Box\Gamma,\psi\Ra\chi$ is kept untouched:

\begin{center}
\begin{math}
\begin{array}{c@{\extracolsep{2cm}}c}
    \AxiomC{$\Phi, \Gamma, \Box \Gamma, \Box \varphi \Rightarrow \varphi$}
    \UnaryInfC{$\Phi, \Box \Gamma \Rightarrow \Box \varphi$}
    \DisplayProof \vspace{1em}
  &
  \AxiomC{$\Phi,\Gamma,\Box\Gamma,\Box\varphi\rightarrow\psi,\Box\varphi\Ra\varphi$}
    \AxiomC{$S$}
    \BinaryInfC{$\Phi,\Box\Gamma,\Box\varphi\rightarrow\psi\Ra\chi$}
    \DisplayProof
    \\
    \AxiomC{$\Phi, \Gamma, \Box \varphi \Rightarrow \varphi$}
    \UnaryInfC{$\Phi, \Box \Gamma \Rightarrow \Box \varphi$}
    \DisplayProof
  &
      \AxiomC{$\Phi,\Gamma,\psi,\Box\varphi\Ra\varphi$}
    \AxiomC{$S$}
    \BinaryInfC{$\Phi,\Box\Gamma,\Box\varphi\rightarrow\psi\Ra\chi$}
    \DisplayProof 
\end{array}
\end{math}
\end{center}

Our results improve on the work of van der Giessen and Iemhoff~\cite{GieIem23}.
First, our new measure ensures that the naive backward proof search strategy for our new calculus terminates.
This is unusual for sequent calculi for provability logics, and especially for intuitionistic provability logics.
Second, we prove direct cut-elimination for $\sys{G4iSLt}$ using a proof technique similar to the \emph{mhd proof technique}~\cite{Brighton16,GorRamShi21}.
Third, all our results are formalised in Coq and can be found here: \url{https://ianshil.github.io/G4iSLT}. We consequently contribute to the rapidly growing literature of formalised proof theory \cite{DAbDawGor21,ChaLimRei19,DawGor10,FeGoo23,Gat22,GorRamShi21,GorShi22,Shi23}.
We also think that our work sheds light on what one might call proof-theoretic meta considerations. Namely, it shows the subtle consequences of rule choices on termination and cut-elimination.

In Section~\ref{sec:preliminaries}, we introduce the preliminaries of $\sys{iSL}$, including our calculus $\sys{G4iSLt}$.
Section~\ref{sec:admrules} presents the admissibility of structural rules in $\sys{G4iSLt}$.
In Section~\ref{sec:termination}, we prove that backward proof search in $\sys{G4iSLt}$ strongly terminates.
Finally, in Section~\ref{sec:cutelimination}, we directly prove cut-admissibility for $\sys{G4iSL}$ using a proof technique similar to the \emph{mhd proof technique}~\cite{Brighton16,GorRamShi21}.

\section{Preliminaries}\label{sec:preliminaries}

In this section we successively present the syntax, axiomatic system, Kripke semantics and sequent calculus for the logic $\sys{iSL}$.

\subsection{Syntax}

Let~$\mathbb{V}=\{p,q,r\dots\}$ be a countably infinite set of propositional variables on which equality is decidable,
that is $\forall p , q \in \mathbb{V}$, we can decide whether $p=q$ or-else $p \neq q$. Modal formulae are defined using BNF notation as below:
\[
  \varphi ::= p\in\mathbb{V} \mid \bot \mid \varphi \land \varphi \mid \varphi
  \lor \varphi \mid \varphi \imp \varphi \mid \Box \varphi
\]
We use the greek letters $\varphi,\psi,\chi,\delta,\dots$ for formulae and $\Gamma,\Delta,\Phi,\Psi\dots$ for multisets of formulae.
We say that $\varphi$ is a \textit{boxed formula} if $\Box$ is its main connective.
For a multiset $\Gamma$, we define the multiset $\Box\Gamma := \{\Box\varphi: \varphi\in\Gamma\}$. By the unboxing of a multiset $\Box\Gamma$ we mean the multiset $\Gamma$.

Following Gor\'{e} et al.~\cite{GorRamShi21,GorShi22}, we encode formulae as an inductive type \lstinline{MPropF} whose base case encodes $\mathbb{V}$ as the type \lstinline{nat} of natural numbers because \lstinline{nat} is countably infinite and equality is decidable on it.
A list of such formulae then has the type \lstinline{list MPropF}. The usual operations on lists ``append'' and ``cons'' are respectively represented by \lstinline!++! and  \lstinline!::! but Coq also allows us to write lists in infix notation using \lstinline{;}. Thus the terms \lstinline{$\varphi$1 :: $\varphi$2 :: $\varphi$3 :: nil} and \lstinline{[$\varphi$1] ++ [$\varphi$2] ++ [$\varphi$3]} and \lstinline{[$\varphi$1 ; $\varphi$2 ; $\varphi$3]} all encode the list $\varphi_1, \varphi_2, \varphi_3$.

We straightforwardly extend Dyckhoff's notion of weight of a formula \cite{Dyc92}, defined for the intuitionistic language, to the modal language.

\begin{definition}
The \textit{weight} $w(\varphi)$ of a formula $\varphi$ is defined as follows:
\begin{center}
\begin{tabular}[c]{r c l}
$w(\bot)=w(p)$ & = & $1$
 \\
$w(\psi\lor\chi)=w(\psi\rightarrow\chi)$ & = & $w(\psi) + w(\chi) + 1$
\\
$w(\psi\land\chi)$ & = & $w(\psi) + w(\chi) + 2$
\\
$w(\Box \psi)$ & = & $w(\psi) + 1$
\\
\end{tabular}
\end{center}
\end{definition}

The main motivation behind this weight is to ensure that $w(\varphi\rightarrow(\psi\rightarrow\chi))<w((\varphi\land\psi)\rightarrow\chi)$, which is crucial to show termination of naive backward proof search on the sequent calculus $\sys{G4ip}$ for intuitionistic logic.

\subsection{Axiomatic systems as consequence relations}

Traditional Hilbert calculi are designed to capture logics as sets of theorems, that is sets of the form $\{\varphi:\;\vdash\varphi\}$. However, when considering logics as consequence relations these systems are inadequate, and notably lead to historical confusions about properties such as the deduction theorem \cite{GorShi20,HakNeg12}.

Generalised Hilbert calculi manipulate expressions $\Gamma\vdash\varphi$, where $\Gamma$ is a set of formulae.
They clearly distinguish between the notion of deducibility from a set of assumptions, versus theoremhood.
They are particularly useful for identifying the appropriate form of deduction theorem holding for a logic~\cite{GorShi20}.
Still, they correspond to traditional Hilbert calculi when restricted to consecutions of the shape $\emptyset\vdash\varphi$, as we do here. Thus, we can connect the generalised Hilbert calculus here to the traditional Hilbert calculus considered by Ardeshir and Mojtahedi~\cite{ArdMoj18}.

The generalised Hilbert calculus $\sys{iSLH}$ for~$\sys{iSL}$, shown in Figure~\ref{fig:Hilbert}, extends the one for intuitionistic modal logic~$\sys{iK}$ with the Strong L{\"{o}}b axiom~$(\Box\varphi\imp\varphi)\imp\varphi$. We write $\Gamma\vdash_{\sys{iSLH}}\varphi$ if $\Gamma\vdash\varphi$ is provable in $\sys{iSLH}$.

Note that if we replace the premise of the rule \rule{Nec} by $\Gamma\vdash\varphi$ we obtain an equivalent calculus. This is implied by the completeness axiom $\varphi\imp\Box\varphi$ and the holding of the deduction theorem in $\sys{iSLH}$ \cite{Gie22}.

\begin{figure}[t]
\begin{center}
\begin{tabular}{c}
\textbf{Axioms}\\
\end{tabular}
\end{center}
\vspace{-0.8cm}

\begin{center}
\begin{tabular}{l l l l}
$A_{1}$ & $\varphi\imp(\psi\imp\varphi)$ & $A_{7}$ & $(\varphi\land\psi)\imp\psi$ \\
$A_{2}$ & $(\varphi\imp(\psi\imp\chi))\imp((\varphi\imp\psi)\imp(\varphi\imp\chi)))$ & $A_{8}$ &  $(\varphi\rightarrow\psi)\rightarrow((\varphi\rightarrow\chi)\rightarrow(\varphi\rightarrow(\psi\land\chi)))$ \\
$A_{3}$ & $\varphi\rightarrow(\varphi\lor\psi)$ & $A_{9}$ & $\bot\rightarrow\varphi$ \\
$A_{4}$ & $\psi\rightarrow(\varphi\lor\psi)$ & $A_{10}$ & $\Box(\varphi\rightarrow\psi)\rightarrow(\Box\varphi\rightarrow\Box\psi)$ \\
$A_{5}$ & $(\varphi\rightarrow\chi)\rightarrow((\psi\rightarrow\chi)\rightarrow((\varphi\lor\psi)\rightarrow\chi))$ & $A_{11}$ & $(\Box\varphi\rightarrow\varphi)\rightarrow\varphi$ \\
$A_{6}$ & $(\varphi\land\psi)\imp\varphi$ & & \\
\end{tabular}
\end{center}
\vspace{-0.7cm}

\begin{center}
\begin{tabular}{c}
\textbf{Rules of Inference}\\
\end{tabular}
\end{center}
\vspace{-0.5cm}

\begin{center}
\begin{tabular}{c@{\hspace{1.5cm}}c}
$
\inferLineSkip=3pt
\infer[\scriptstyle\rule{Ax}]{\Gamma\vdash\varphi}{\varphi\text{ is an instance of an axiom }}
$ & 
$
\inferLineSkip=3pt
\infer[\scriptstyle\rule{El}]{\Gamma\vdash\varphi}{\varphi\in\Gamma}
$ \\
 & \\
$
\inferLineSkip=3pt
\infer[\scriptstyle\rule{Nec}]{\Gamma\vdash\Box\varphi}{\emptyset\vdash\varphi}
$ &
$
\inferLineSkip=3pt
\infer[\scriptstyle\rule{MP}]{\Gamma\vdash\psi}{
	\Gamma\vdash\varphi
	&
	\Gamma\vdash\varphi\rightarrow\psi}
$
\end{tabular}
\end{center}
\vspace{-0.4cm}
\caption{Generalised Hilbert calculus $\sys{iSLH}$ for $\sys{iSL}$}
\label{fig:Hilbert}
\end{figure}

\subsection{Kripke semantics}

We now present the Kripke semantics for $\sys{iSL}$ \cite{Lit14,ArdMoj18} to notably prove soundness of our sequent calculus $\sys{G4iSLt}$, and explain its rules \rule{SLtR} and \rule{$\Box\!\imp$L}. 

The Kripke semantics of $\sys{iSL}$ is a restriction of the Kripke semantics for intuitionistic modal logics. More precisely, the semantic interpretation of connectives is preserved, but the class of models is restricted. The models for this logic are defined below, where for a set $W$, we write $\mathcal{P}(W)$ for the set of all subsets of $W$.

\begin{definition}\label{model}
A \emph{Kripke model}~$\mathcal M$ for $\sys{iSL}$ is a tuple $(W,\leq,R,I)$, where $W$ is a non-empty set (of possible worlds), both $\leq$ (the intuitionistic relation) and $R$ (the modal relation) are subsets of $W\times W$, and $I:\mathbb V\rightarrow\mathcal P(W)$, which satisfies the following:
$\leq$ is reflexive and transitive;
$R$ is transitive and converse well-founded;
$(\leq \circ R)\,\subseteq\,R$ where ``$\circ$''
is relational composition; 
$R\,\subseteq\,\leq$;
and for all $p\in\mathbb V$ and $w,v\in W$, if $w\leq v$ and $w\in I(p)$ then $v\in I(p)$.
\end{definition}

Note the peculiarity of the models for $\sys{iSL}$: $R\,\subseteq\,\leq$, that is the modal relation is a subset of the intuitionistic relation. We recall the standard definition of forcing for intuitionistic modal logics, and show that persistence holds.

\begin{definition}\label{DefForcMod}
  Given a Kripke model~$\mathcal M=(W,\leq,R,I)$, we define the forcing relation as follows, where $v \geq w$ is just $w \leq v$:
\begin{center}
\begin{tabular}{l @{\extracolsep{1em}} c l}
$\mathcal M,w\Vdash p$ & if & $w\in I(p)$\\
$\mathcal M,w\Vdash\bot$ & & never \\
$\mathcal M,w\Vdash\varphi\land\psi$ & if & $\mathcal M,w\Vdash\varphi$ and $\mathcal M,w\Vdash\psi$\\
$\mathcal M,w\Vdash\varphi\lor\psi$ & if & $\mathcal M,w\Vdash\varphi$ or $\mathcal M,w\Vdash\psi$\\
$\mathcal M,w\Vdash\varphi\rightarrow\psi$ & if & $\forall v \geq w.\; \mathcal M,v\Vdash\varphi$ implies $\mathcal M,v\Vdash\psi$\\
$\mathcal M,w\Vdash\Box\varphi$ & if & $\forall v \in W.\; wRv$ implies $\mathcal M,v\Vdash\varphi$\\
\end{tabular}
\end{center}
Local consequence is as below where $\mathcal M, w \Vdash \Gamma$ means $\forall \varphi \in \Gamma, \mathcal M, w \Vdash \varphi$:
\begin{center}
\begin{tabular}{l@{\hspace{1cm}}c@{\hspace{1cm}}l}
$\Gamma\models\varphi$ & iff & $\forall\mathcal M.\forall w.\,(\mathcal M,w\Vdash\Gamma\;\;\;\text{implies}\;\;\;\mathcal M,w\Vdash\varphi)$\\
\end{tabular}
\end{center}
\end{definition}

\begin{lemma}[Persistence]\label{lem:pers}
For any model $\mathcal M=(W,\leq,R,I)$, formula $\varphi$ and points $w,v\in W$, if $w\leq v$ and $\mathcal M,w\Vdash\varphi$ then $\mathcal M,v\Vdash\varphi$.
\end{lemma}

Interestingly, as $\sys{iSL}$ satisfies the finite model property \cite{VisserZoethout18} it can also be characterised by the class of \emph{finite} frames where $R$ is transitive and \emph{irreflexive}.

\subsection{Sequent calculus}

A \textit{sequent} is a pair of a finite multiset $\Gamma$ of formulae and a formula $\varphi$, denoted $\Gamma\Ra \varphi$.
For a sequent $\Gamma\Ra\varphi$ we call $\Gamma$ the \textit{antecedent} of the sequent and $\varphi$ the \textit{consequent} of the sequent.
For multisets~$\Gamma$ and~$\Delta$, the multiset sum $\Gamma \uplus \Delta$ is the multiset whose multiplicity (at each formula) is a sum of the multiplicities of~$\Gamma$ and~$\Delta$. We write $\Gamma,\Delta$ to mean $\Gamma \uplus \Delta$. For a formula~$\varphi$, we write $\varphi,\Gamma$ and $\Gamma,\varphi$ to mean $\{\varphi\} \uplus \Gamma$. From the formalisation perspective, a pair of a list of formulae \lstinline{(list MPropF)} and a formula \lstinline{MPropF} has type \lstinline{(list MPropF) * MPropF}, using the Coq notation \lstinline{*} for forming pairs. The latter is the type we give to sequents in our formalisation, for which we use the macro \lstinline{Seq}. Thus the sequent $\varphi_1, \varphi_2, \varphi_3 \Ra \psi$ is encoded by the term \lstinline{[$\varphi$1 ; $\varphi$2 ; $\varphi$3] * $\psi$}, which itself can also be written as the pair \lstinline{([$\varphi$1 ; $\varphi$2 ; $\varphi$3], $\psi$)}.
Note that \lstinline{[$\varphi$1 ; $\varphi$2 ; $\varphi$3] * $\psi$} is different from \lstinline{[$\varphi$2 ; $\varphi$1 ; $\varphi$3] * $\psi$}
since the order of the elements is crucial, so our lists do not capture multisets (yet).

A \textit{sequent calculus} consists of a finite set of \textit{sequent rule schemas}. 
Each rule schema consists of a conclusion sequent schema and some number of premise sequent schemas. 
A rule schema with zero premise schemas is called an initial rule.
The conclusion and premises are built in the usual way from propositional-variables, formula-variables and multiset-variables. 
A \textit{rule instance} is obtained by uniformly instantiating every variable in the rule schema with a concrete object of that type. 
This is the standard definition from structural proof theory.

\begin{definition}[Derivation/Proof]\label{def-sd}
A \emph{\der} of a sequent~$S$ in the sequent calculus~$\sys{C}$ is a finite tree of sequents such that
(i) the root node is~$S$; and
(ii) each interior node and its direct children are the conclusion and premise(s) of a rule instance in~$\sys{C}$.
A \emph{\prf} is a \der~where every leaf is the conclusion of an instance of an initial rule.
\end{definition}

Note that we explicitly define the notion of a derivation as an object rather than define the notion of derivability, as is done in some papers. We do so as we want to create a ``deep'' embedding of such derivations into Coq \cite{DawGor10}.

In what follows, it should be clear from context whether the word ``proof'' refers to the object defined in Definition \ref{def-sd}, or to the meta-level notion. We say that a sequent is \emph{\prov}~in~$\sys{G4iSLt}$ if it has a \prf~in~$\sys{G4iSLt}$.
We elide the details of the encodings of sequent rules and derivations, as these can be found elsewhere~\cite{DAbDawGor21,Shi23}.
We define a predicate \lstinline{G4iSLt_prv} on sequents to encode \textit{provability}~in $\sys{G4iSLt}$.
Our encodings rely on the type \lstinline{Type}, which bears computational content, unlike \lstinline{Prop}, and is crucially compatible with the extraction function of Coq.

Before presenting our calculus, we recall standard notions from proof theory.

\begin{definition}[Height]
For any \der~$\delta$, its \emph{height} $h(\delta)$ is the maximum number of nodes on a path from root to leaf.
\end{definition}

\begin{definition}[Admissibility, Invertibility, Height-Preservation]\label{GenDefAdmInv}
Let $\mathsf R$ be a rule schema with premises $S_0,\dots,S_n$ and conclusion $S$. We say that $\mathsf R$ is:
\begin{description}
\item[\rm admissible:] if for every instance of $\mathsf R$, the instance of $S$ is provable whenever the instances of $S_1,\dots,S_n$ are all provable; 
	
\item[\rm invertible:] if for every instance of $\mathsf R$, the instances of $S_1,\dots,S_n$ are all provable whenever the instance of $S$ is provable;
	
\item[\rm height-preserving admissible:] if for every instance of $\mathsf R$, if there are proofs $\pi_0,\dots,$ $\pi_n$ of the instances of $S_0,\dots,S_n$ then there is a proof $\pi$ of the instance of $S$ such that $h(\pi)\leq h(\pi_i)$ for some $0\leq i\leq n$; 

\item[\rm height-preserving invertible:] if for every instance of $\mathsf R$, if $\pi$ is a proof of the instance of $S$ then there are proofs $\pi_0,\dots,\pi_n$ of the instances of $S_0,\dots,S_n$ such that $h(\pi_i)\leq h(\pi)$ for all $0\leq i\leq n$.
\end{description}
\end{definition}

The sequent calculus~$\sys{G4iSLt}$ is given in Figure~\ref{fig:iseq-pc}. When defining rules we put the label naming of the rule on the left of the horizontal line, while the label appears on the right of the line in \textit{instances} of rules.

\begin{figure}[t]
\centering
{\small
  \begin{tabular}{l@{\hspace{1cm}}ll}
\AxiomC{}
\LeftLabel{$\scriptstyle\rule{$\bot$L}$}
\UnaryInfC{$\bot, \Gamma\Ra\chi$}
\DisplayProof
&
\AxiomC{}
\LeftLabel{$\scriptstyle\rule{IdP}$}
\UnaryInfC{$\Gamma,p\Ra p$}
\DisplayProof

\\[0.6cm]

\AxiomC{$\Gamma,\varphi,\psi\Ra\chi$}
\LeftLabel{$\scriptstyle\rule{$\land$L}$}
\UnaryInfC{$\Gamma,\varphi\land\psi\Ra\chi$}
\DisplayProof
&
\AxiomC{$\Gamma\Ra\varphi$}
\AxiomC{$\Gamma\Ra\psi$}
\LeftLabel{$\scriptstyle\rule{$\land$R}$}
\BinaryInfC{$\Gamma\Ra\varphi\land\psi$}
\DisplayProof

\\[0.6cm]

\AxiomC{$\Gamma,\varphi\Ra\chi$}
\AxiomC{$\Gamma,\psi\Ra\chi$}
\LeftLabel{$\scriptstyle\rule{$\lor$L}$}
\BinaryInfC{$\Gamma,\varphi\lor\psi\Ra\chi$}
\DisplayProof
&
\AxiomC{$\Gamma\Ra\varphi_i$}
\LeftLabel{$\scriptstyle\rule{$\lor$R$_i$}$}
\RightLabel{$(i \in \{1,2\})$}
\UnaryInfC{$\Gamma\Ra\varphi_1\lor\varphi_2$}
\DisplayProof

\\[0.6cm]
\AxiomC{$\Gamma,p,\varphi\Ra\chi$}
\LeftLabel{$\scriptstyle\rule{$p\!\rightarrow$L}$}
\UnaryInfC{$\Gamma,p,p\rightarrow\varphi\Ra\chi$}
\DisplayProof
&
\AxiomC{$\Gamma,\varphi\Ra\psi$}
\LeftLabel{$\scriptstyle\rule{$\rightarrow$R}$}
\UnaryInfC{$\Gamma\Ra \varphi\rightarrow\psi$}
\DisplayProof

\\[0.6cm]

\AxiomC{$\Phi,\Gamma,\psi,\Box\varphi\Ra\varphi$}
\AxiomC{$\Phi,\Box\Gamma,\psi\Ra\chi$}
\LeftLabel{$\scriptstyle\rule{$\Box\!\rightarrow$L}$}
\BinaryInfC{$\Phi,\Box\Gamma,\Box\varphi\rightarrow\psi\Ra\chi$}
\DisplayProof 
&
\AxiomC{$\Phi,\Gamma,\Box\varphi\Ra\varphi$}
\LeftLabel{$\scriptstyle\rule{SLtR}$}
\UnaryInfC{$\Phi,\Box\Gamma\Ra \Box\varphi$}
\DisplayProof

  \\[0.6cm]
\AxiomC{$\Gamma,\varphi\rightarrow (\psi\rightarrow\chi)\Ra\delta$}
\LeftLabel{$\scriptstyle\rule{$\land\!\rightarrow$L}$}
\UnaryInfC{$\Gamma,(\varphi\land\psi)\rightarrow\chi\Ra\delta$}
\DisplayProof
&
\AxiomC{$\Gamma,\varphi\rightarrow\chi,\psi\rightarrow\chi\Ra\delta$}
\LeftLabel{$\scriptstyle\rule{$\lor\!\rightarrow$L}$}
\UnaryInfC{$\Gamma,(\varphi\lor\psi)\rightarrow\chi\Ra\delta$}
\DisplayProof

\\[0.6cm]
\multicolumn{2}{c}{
\AxiomC{$\Gamma,\psi\rightarrow\chi\Ra \varphi\rightarrow\psi$}
\AxiomC{$\Gamma,\chi\Ra\delta$}
\LeftLabel{$\scriptstyle\rule{$\rightarrow\rightarrow$L}$}
\BinaryInfC{$\Gamma,(\varphi\rightarrow\psi)\rightarrow\chi\Ra\delta$}
\DisplayProof
}

\end{tabular}
}
\caption{The sequent calculus~$\sys{G4iSLt}$, where $\Phi$ contains no boxed formula.}
  \label{fig:iseq-pc}
\end{figure}

In $\rule{IdP}$, a propositional variable instantiating the featured occurrences of~$p$ is principal.
In a rule instance of~\rule{$\land$R}, \rule{$\land$L}, \rule{$\lor$R$_i$}, \rule{$\lor$L} or~\rule{$\imp$R}, the \textit{principal formula} of that instance is defined as usual. 
In a rule instance of~\rule{$p\!\rightarrow$L}, both a propositional variable instantiating $p$ and the formula instantiating the featured~$p\imp \varphi$ are principal formulae of that instance.
In a rule instance of~\rule{$\land\!\rightarrow$L}, \rule{$\lor\!\rightarrow$L}, \rule{$\rightarrow\rightarrow$L} or \rule{$\Box\!\rightarrow$L}, the formula instantiating respectively $(\varphi\land\psi)\rightarrow\chi$, $(\varphi\lor\psi)\rightarrow\chi$, $(\varphi\rightarrow\psi)\rightarrow\chi$ or $\Box \varphi\rightarrow\psi$ is the principal formula of that instance. 
In a rule instance of \rule{SLtR} or \rule{$\Box\!\imp$L}, $\Box\varphi$
is called the \textit{diagonal formula}~\cite{SambinValentini82}.

The non-modal rules are taken from the calculus for $\sys{IPC}$ for which backward proof search strongly terminates \cite{Dyc92}.
Keypoint is that the usual intuitionistic left implication rule is replaced by four implication rules depending on the main connective in the antecedent of the principal formula, in such a way that each premise is less complex than the conclusion. In particular, when considering the rule $\rule{$\rightarrow\rightarrow$L}$, an application of the ``regular'' left implication rule yields the more complex left premise $\Gamma,(\varphi\rightarrow\psi)\rightarrow\chi\Ra\varphi\rightarrow\psi$, which is (semantically) equivalent to the simpler left premise stated in rule~$\rule{$\rightarrow\rightarrow$L}$.

We proceed to give semantic intuitions for the rules \rule{SLtR} and \rule{$\Box\!\imp$L}.

The \rule{SLtR} rule has similarities with the rule \rule{GLR} (shown below) from sequent calculi for provability logics such as $\sys{GL}$, but with two major differences: (1) the non-boxed formulae $\Phi$ in the antecedent of the sequent are preserved from conclusion to premise in \rule{SLtR}, while they are deleted in \rule{GLR}; and (2) the formulae in $\Box\Gamma$ are not preserved upwards in \rule{SLtR}, while they are in \rule{GLR}.
\begin{center}
\begin{tabular}{c@{\hspace{2cm}}c}
$
\infer[\scriptstyle\rule{SLtR}]{\Phi,\Box\Gamma\Ra \Box\varphi}{\Phi,\Gamma,\Box\varphi\Ra\varphi}
$
&
$
\infer[\scriptstyle\rule{GLR}]{\Phi,\Box\Gamma\Ra \Box\varphi}{\Gamma,\Box\Gamma,\Box\varphi\Ra\varphi}
$ \\
\end{tabular}
\end{center}
From a backward proof search perspective, both rules correspond, semantically, to a ``modal jump'' from a point $w$ which falsifies the conclusion $\Phi,\Box\Gamma\Ra \Box\varphi$ to a modal successor $v$ which forces $\Gamma$ but falsifies the succedent $\varphi$ of the premise.
The underlying relation $R$ in both logics is transitive and converse well-founded. Using converse well-foundedness we can assume that~$v$ is the last modal successor making $\varphi$ false, thus $v$ forces $\Box\varphi$ in both logics. Transitivity implies that~$v$ forces~$\Box\Gamma$ in both logics, so all its successors force $\Gamma$. But, in $\sys{iSL}$, the underlying relation $R$ is also persistent so $v$ also forces $\Phi$ in $\sys{iSL}$, but not in $\sys{GL}$, thus explaining difference (1). Thanks to persistence, $v$ forcing $\Gamma$ implies that all its successors force $\Gamma$, meaning that $v$ forces $\Box\Gamma$ already, thus explaining difference~(2).

The two premises of \rule{$\Box\!\imp$L} capture how $\Box\varphi\imp\psi$ in the antecedent of the conclusion can be true.
The simple case is when $\psi$ is true, which corresponds to the right premise. The more complicated case is when $\psi$ is not true, implying that $\Box\varphi$ must also be not true.
Now, $\Box\varphi$ true semantically means that $\varphi$ is true in all modal successors, hence $\Box\varphi$ not true means that $\varphi$ is not true in a modal successor.
But converse well-foundedness implies the existence of a last modal successor where $\varphi$ is not true, with all its modal successors making $\varphi$ true.
The left premise corresponds to this last modal successor, as it encodes that $\varphi$ is not true but $\Box\varphi$ is true.
Moreover, this last modal successor is also an intuitionistic successor as $R\,\subseteq\,\leq$. By persistence, this last successor must also make $\Box\varphi\imp\psi$ true.
But then, a simple modus ponens on $\Box\varphi$ and $\Box\varphi\imp\psi$ gives us $\psi$.
  
Finally, we show that $\sys{G4iSLt}$ indeed captures the set of theorems of $\sys{iSL}$.

\begin{theorem}
For all $\varphi$ we have:
$\emptyset\vdash_{\sys{iSLH}}\varphi$ iff $\Ra\varphi$~is \prov~in~$\sys{G4iSLt}$.
\end{theorem}

\begin{proof}
We proved in Coq the two following results.
\begin{center}
\begin{tabular}{c@{\hspace{0.3cm}} l@{\hspace{0.5cm}}c@{\hspace{0.5cm}}l}
(1) & $\Gamma\vdash_{\sys{iSLH}}\varphi$ & implies & there exists a finite $\Gamma'\subseteq\Gamma$ s.t. \\
 &  &  & $\Gamma'\Ra\varphi$ is provable in $\sys{G4iSLt}$ \\
 
(2) & $\Gamma\Ra\varphi$ is provable in $\sys{G4iSLt}$ & implies & $\Gamma\models\varphi$\\
\end{tabular}
\end{center}
The result (1), which relies on the admissibility of cut (Theorem \ref{thm-CE-G4iSLt}), shows that $\sys{G4iSLt}$ is (strongly) complete with respect to $\sys{iSLH}$ and gives us the left-to-right direction of our theorem.
The other direction involves the soundness of $\sys{G4iSLt}$ w.r.t.~the local consequence shown in (2), as well as the (non-formalised) result of (weak) completeness of $\sys{iSLH}$ w.r.t.~the local consequence obtained by Ardeshir and Mojtahedi \cite{ArdMoj18}.
\qed
\end{proof}

\section{Admissible rules in~$\sys{G4iSLt}$}\label{sec:admrules}

This section aims at showing that the contraction rule is admissible. To do so, it follows the work developed by Gor\'{e} and Shillito \cite{GorShi22} on the sequent calculus $\sys{GL4ip}$ for the intuitionistic provability logic $\sys{iGL}$, which extends itself on the work of Dyckhoff and Negri \cite{DycNeg00} on $\sys{G4ip}$. Most of the overall structure of the argument is the same as for the case of $\sys{GL4ip}$, except for the crucial and typical \emph{left-unboxing rule} \rule{$\boxtimes$}, shown to be height-preserving admissible.

Most of the results of this section are proven by inductions on the weight of formulae and/or height of derivations. We omit the Coq encodings for brevity.

\begin{lemma}[Height-preserving invertibility of rules]\label{lem:inv_rules}
 The rules $\rule{$\land$R},$ $\rule{$\land$L}$, $\rule{$\lor$L},\rule{$\rightarrow$R}, \rule{$p\!\rightarrow$L},$ $\rule{$\land\!\rightarrow$L},\rule{$\lor\!\rightarrow$L}$ are height-preserving invertible.
\end{lemma}

We present height-preserving admissible and admissible rules in Figure \ref{fig:adm}.

\begin{figure}[t]
\begin{center}
\begin{tabular}{c}
\textbf{Height-preserving admissible rules}
\end{tabular}
\end{center}

\begin{center}
\begin{tabular}{c@{\hspace{1cm}}c@{\hspace{1cm}}c}
\AxiomC{$\Gamma_0,\Gamma_3,\Gamma_2,\Gamma_1,\Gamma_4\Ra\chi$}
\LeftLabel{$\scriptstyle\rule{Exc}$}
\UnaryInfC{$\Gamma_0,\Gamma_1,\Gamma_2,\Gamma_3,\Gamma_4\Ra\chi$}
\DisplayProof
& 
\AxiomC{$\Gamma\Ra\chi$}
\LeftLabel{$\scriptstyle\rule{Wkn}$}
\UnaryInfC{$\Gamma,\varphi\Ra\chi$}
\DisplayProof
&
\AxiomC{$\Phi,\Box\Gamma\Ra\chi$}
\LeftLabel{$\scriptstyle\rule{$\boxtimes$}$}
\UnaryInfC{$\Phi,\Gamma\Ra\chi$}
\DisplayProof
\end{tabular}
\end{center}

\begin{center}
\begin{tabular}{c@{\hspace{1.5cm}}c}
\AxiomC{$\Phi,\Box\Gamma,\Box\varphi\rightarrow\psi\Ra\chi$}
\LeftLabel{$\scriptstyle\rule{$\Box\!\imp$LIR}$}
\UnaryInfC{$\Phi,\Box\Gamma,\psi\Ra\chi$}
\DisplayProof
&
\AxiomC{$\Gamma,(\varphi\rightarrow\psi)\rightarrow\chi\Ra\delta$}
\LeftLabel{$\scriptstyle\rule{$\imp\imp$LIR}$}
\UnaryInfC{$\Gamma,\chi\Ra\delta$}
\DisplayProof
\end{tabular}
\end{center}
 
\begin{center}
\begin{tabular}{c}
\textbf{Admissible rules}
\end{tabular}
\end{center}

\begin{center}
\begin{tabular}{c@{\hspace{1cm}}c}
\AxiomC{}
\LeftLabel{$\scriptstyle\rule{Id}$}
\UnaryInfC{$\varphi,\Gamma\Ra\varphi$}
\DisplayProof
&
\AxiomC{$\Gamma\Ra\varphi$}
\AxiomC{$\Gamma,\psi\Ra\chi$}
\LeftLabel{$\scriptstyle\rule{$\rightarrow$L}$}
\BinaryInfC{$\Gamma,\varphi\rightarrow\psi\Ra\chi$}
\DisplayProof
\end{tabular}
\end{center}

\begin{center}
\begin{tabular}{c@{\hspace{1cm}}c}
\AxiomC{$\Gamma,(\varphi\rightarrow\psi)\rightarrow\delta\Ra\chi$}
\LeftLabel{$\scriptstyle\rule{$\imp\imp$LIL}$}
\UnaryInfC{$\Gamma,\varphi,\psi\rightarrow\delta,\psi\rightarrow\delta\Ra\chi$}
\DisplayProof
&
\AxiomC{$\varphi,\varphi,\Gamma\Ra\chi$}
\LeftLabel{$\scriptstyle\rule{Ctr}$}
\UnaryInfC{$\varphi,\Gamma\Ra\chi$}
\DisplayProof
\end{tabular}
\end{center}
\caption{Height-preserving admissible and admissible rules in $\sys{G4iSLt}$.}
\label{fig:adm}
\end{figure}

The structural rules of weakening \rule{Wkn}, contraction \rule{Ctr} and exchange \rule{Exc}, are all (at least) admissible. The presence of the latter may be surprising, as the sequents we use are based on multisets. However, as mentioned earlier, our formalisation encodes sequents using lists and not multisets. So, the formal proof of the height-preserving admissibility of \rule{Exc} shows that list-sequents of our formalisation mimic multiset-sequents of the pen-and-paper definition. In fact, we designed the formalisation of $\sys{G4iSLt}$ so that it admits exchange \cite{GorShi22}. 

The rule \rule{$\boxtimes$} is quite typical of the logic $\sys{iSL}$, as it reflects one of its theorems: the completeness axiom $\varphi\imp\Box\varphi$. Indeed, this axiom implies that $\Gamma$ entails~$\Box\Gamma$, allowing the replacement of~$\Box\Gamma$ by $\Gamma$ in the antecedent of a provable sequent while preserving provability. The height-preserving admissibility of \rule{$\boxtimes$} is crucially used in many places, notably Lemma \ref{lem:inv_rules} and the admissibility of cut.

The height-preserving admissibility of \rule{$\Box\!\rightarrow$LIR} and \rule{$\rightarrow\rightarrow$LIR} shows height-preserving invertibility in the right premise of the rules \rule{$\Box\!\rightarrow$L} and \rule{$\imp\imp$L}.

The admissible rule \rule{$\rightarrow$L} is the traditional left-implication rule. 
We use this rule to prove the admissibility of \rule{$\rightarrow\rightarrow$LIL}, resembling the invertibility in the left premise of \rule{$\rightarrow\rightarrow$L}.
In turn, \rule{$\rightarrow\rightarrow$LIL} is crucial in the admissibility of \rule{Ctr}.

In the following section we introduce a measure on sequents which we use to show that the naive backward proof search strategy for $\sys{G4iSLt}$ terminates. This measure could thus be used to derive the notion of maximum height of derivations (mhd) for a sequent, as was done in previous works \cite{GorRamShi21,GorShi22}. There, the mhd measure was used as secondary induction measure in the proof of admissibility of cut. Here, we simply use the termination measure instead.

\section{Naive backward proof search terminates}\label{sec:termination}

Sequent calculi enjoying cut-elimination can often be used to decide whether a given formula $\varphi$ is deducible from a given set of assumptions $\Gamma$ by strategically applying the rules ``backwards'' from the end-sequent $\Gamma \Ra \varphi$.
To obtain a decision procedure, we require a backward proof search strategy which terminates and is complete, 
i.e.~which provides a proof for any sequent provable in the calculus.

But often, terminating complete strategies necessitate a ``loop check'' mechanism, that stops the search if the same sequent appears twice on a branch.
For example, the sequent calculus $\sys{LJ}$, for propositional intuitionistic logic, only has a strategy with loop check as terminating complete strategy.
The termination of these strategies is messy to reason about, as in most cases their unguarded version is not terminating and results in proof trees with infinite branches.

While some calculi have terminating complete strategies without loop checks, like $\sys{GLS}$ for $\sys{GL}$ \cite{GorRamShi21} and $\sys{GL4ip}$ for $\sys{iGL}$ \cite{GieIem21}, we consider a stronger kind of calculus: calculi with \emph{strongly terminating} backward proof search, such as $\sys{G4ip}$ for intuitionistic propositional logic \cite{Dyc16}.
Backward proof search for a sequent calculus is strongly terminating if and only if \emph{all} backward proof search strategies for this calculus, complete or not, terminate.
This characterisation has other equivalent forms: 
(1) the naive backward proof search strategy terminates,
and (2) there is a well-founded ordering on sequents 
decreasing upwards in all the rules of the calculus.
In contrast, backward proof search is \emph{weakly} terminating if and only if \emph{there is} a terminating complete strategy for this calculus.

In this section we show that backward proof search for $\sys{G4iSLt}$ is strongly terminating. More precisely, we show that the naive strategy terminates.
To do this, we need two ingredients: (1) a locally defined measure on sequents, and (2) a well-founded order making this measure decrease upwards in the rules of~$\sys{G4iSLt}$.

\subsection{Shortlex: a well-founded order on \lstinline{list} $\mathbb N$}

We define the shortlex order, which is a well-founded order on \lstinline{list} $\mathbb N$, i.e.~the set of all lists of natural numbers.

In the following, we use $<$ to mean the usual ordering on natural numbers. Let us recall the definition of the lexicographic order on lists of natural numbers.

\begin{definition}[Lexicographic order]
Let $n\in\mathbb N$. We define the lexicographic order $<_{lex}^n$ on lists of natural numbers of length $n$. For two lists of natural numbers $[m_1;\cdots;m_n]$ and $[k_1;\cdots;k_n]$, we write $[m_1;\cdots;m_n]<_{lex}^{n}[k_1;\cdots;k_n]$ if there is a $1\!\leq\! j \!\leq\! n$ such that: (1) $m_p=k_p$ for all $1\leq p<j$, and (2) $m_j<k_j$.
\end{definition}

Note that as $<$ is a well-founded order, then $<_{lex}^{n}$ is also well-founded \cite{Pau86}. Finally, we define the shortlex order, also called \textit{breadth-first} \cite{LarMat19} or \textit{length-lexicographic} order, over lists of natural numbers (viewed as $n$-tuples).

\begin{definition}[Shortlex order]
  The shortlex order over lists of natural numbers, noted $\slex$, is defined as follows.
  For two lists $l_0$ and $l_1$ of natural numbers,
  we say that $l_0\slex l_1$ whenever one of the following conditions is satisfied:
\begin{enumerate}
\item $length(l_0)<length(l_1)$ ; 
\item $length(l_0)=length(l_1)=n$ and $l_0<_{lex}^{n}l_1$.
\end{enumerate}
\end{definition}

Intuitively, the shortlex order is ordering lists according to their length and follows the lexicographic order whenever length does not discriminate. Note that on top of being well-founded, $\slex$ is obviously transitive.

\subsection{A (\lstinline{list} $\mathbb N$)-measure on sequents}

We proceed to attach to each sequent $\Gamma\Ra \chi$ a ``measure'' $\Theta(\Gamma\Ra \chi)$ which is a (finite) list of natural numbers, i.e.~of type \lstinline{list} $\mathbb N$. For simplicity, in the following we consider a fixed sequent $\Gamma\Ra \chi$ for which we define the measure.

To introduce our measure, we first wish to explain why the measure used for $\sys{GL4ip}$ \cite{GorShi22}, acting as a substitute of the Dershowitz-Manna order \cite{DerZoh79} considered in Dyckhoff's article on $\sys{G4ip}$ \cite{Dyc92}, does not work for our purpose. The explanation of this failure justifies the modification we made to obtain the measure for $\sys{G4iSLt}$.

The intuition behind the measure for $\sys{GL4ip}$ and $\sys{G4ip}$ is the following: for a multiset we create an ordered list of counters for each weight of occurrences of formulae of this weight.
For more details, take a finite multiset of formulae~$\Delta$. 
As it is finite, it contains a \emph{topmost} formula of maximal weight $n$. 
We can create a list of length $n$ such that at each position $m$ in the list (counting from right to left) for $1\leq m\leq n$, we find the number of occurrences in $\Delta$ of \emph{topmost} formulae of weight $m$. 
Such a list gives the count of occurrences in $\Delta$ of formulae of weight~$n$ in its leftmost (i.e.~$n$-th) component, then of occurrences of formulae of weight~$n-1$ in the next (i.e.~$(n-1)$-th) component, and so on until we reach~$1$. 

The measure for $\sys{GL4ip}$ and $\sys{G4ip}$ consisted in attaching to $\Gamma\Ra\chi$ the list obtained by applying the above procedure on the multiset $\Gamma\uplus\{\chi\}$. Call this function $\Theta_{fail}$. This measure fails to show termination of the naive strategy for $\sys{G4iSLt}$, as it does not decrease upwards in the following application of \rule{SLtR}.

$$
\infer[\scriptstyle\rule{SLtR}]{\Ra\Box p}{\Box p\Ra p}
$$

We have that $\Theta_{fail}(\Ra\Box p)=[1,0]$ because $\Box p$ is the formula of maximum weight $2$, and it is the only formula with this weight occurring in the list, while no formula of weight $1$ appears in $\Ra\Box p$.
In addition to that, we have that $\Theta_{fail}(\Box p\Ra p)=[1,1]$.
Consequently, we obtain $\Theta_{fail}(\Ra\Box p)\slex\Theta_{fail}(\Box p\Ra p)$: the measure increased upwards.
So, the measure used for $\sys{GL4ip}$ and $\sys{G4ip}$ cannot be used here. We need to define another one.

With enough scrutinising, one can notice that in $\sys{G4iSLt}$ the principal box of a boxed formula in the antecedent of a sequent is a ``deadweight''.
More precisely, once a formula $\Box\varphi$ is in the antecedent of a sequent, only two things can happen to its outermost box: it is either deleted (via the modal rule \rule{SLtR} or \rule{$\Box\!\imp$L}), or else it is preserved (through all other rules).
Intuitively, this observation suggests that boxed formulae in the antecedent are destined to be unboxed eventually in the upward application of rules, without having any other effect.

Consequently, as the top-level boxes in the antecedent of a sequent are deadweights, we can think about unboxing the antecedent of $\Gamma\Ra \chi$ before applying the procedure described above.
This is precisely what we do: if $\Gamma$ is of the shape $\Gamma_0,\Box\Gamma_1$ with no boxed formula in $\Gamma_0$, we define $\Theta(\Gamma\Ra\chi)$ to be the list of natural numbers obtained via the above machinery applied on the multiset $\Gamma_0\uplus\Gamma_1\uplus\{\chi\}$.

For example, to compute $\Theta(\Box(p\land q), p\lor q\Ra q\imp p)$, we first unbox the antecedent of this sequent by transforming
  $\Box(p\land q)$ into $p\land q$ to obtain the multiset $\{p\land q, p\lor q, q\imp p\}$.
Because $p\land q$ is the only formula of maximum weight four, our list of length four begins with 1.
Since both $p\lor q$ and $q\imp p$ are of weight three, the second element is 2.
Finally, since there are no formulae of weights two and one, we obtain $\Theta(\Box(p\land q), p\lor q\Ra q\imp p) = [1,2,0,0]$. Following this explanation, observe that the issue we faced with $\Ra\Box p$ and $\Box p\Ra p$ is now fixed: we first unbox $\Box p$ in $\Box p\Ra p$, hence $\Theta(\Box p\Ra p)=[2]\slex[1,0]=\Theta(\Ra\Box p)$.

Two things need to be noted about such lists. First, if no topmost occurrence of a formula is of weight $1\leq k\leq n$, then a $0$ appears in position $k$ in the list. This is the case for the weight $2$ in the 
last example above. 
Second, as no formula is of weight $0$ we do not dedicate a position for this particular weight in our list.

\subsection{Every rule of $\sys{G4iSLt}$ reduces $\Theta$ upwards}

We obtain the sought after result about our measure $\Theta$: it decreases upwards through the rules of $\sys{G4iSLt}$ on the $\slex$ ordering (see the \hyperref[appendix]{Appendix} for a proof): 

\begin{lemma}\label{lem:meas-decr-rule}
  For all sequents $S_0, S_1,...,S_n$ and for all $1\leq i\leq n$,
  if there is an instance of a rule $r$ of $\sys{G4iSLt}$ of the form below, then $\Theta(S_i)\slex\Theta(S_0)$:
$$
\infer[r]{S_0}{S_1 & \dots & S_n}
$$
\end{lemma}

Clearly, this result implies that the naive strategy for $\sys{G4iSLt}$ terminates: any rule application makes the measure decrease on $\slex$, ensuring termination via well-foundedness of $\slex$. Thus, backward proof search is strongly terminating.

Moreover, this lemma is quite crucial in the proof of admissibility of cut: as we use $\Theta(\Gamma\Ra\chi)$ as secondary induction measure (through well-foundedness of $\slex$) there, we know that we can apply the secondary induction hypothesis on any sequent $S$ which is a premise of $\Gamma\Ra\chi$ through a rule, as $\Theta(S)\slex\Theta(\Gamma\Ra\chi)$.

\section{Cut-elimination for~$\sys{G4iSLt}$}\label{sec:cutelimination}

To reach cut-elimination, our main theorem, we first state and prove the admissibility of the cut rule in a direct and purely syntactic way. More precisely, we prove that the \textit{additive}-cut rule, with \emph{cut formula} $\varphi$, is admissible. This statement and its formalisation are given below, where $\Gamma$ is encoded as \lstinline{$\Gamma$0++$\Gamma$1}.

\begin{theorem}[Admissibility of additive-cut]\label{thm-CE-G4iSLt}
The additive cut rule below is admissible in $\sys{G4iSLt}$.
$$
\infer[\scriptstyle\rule{Cut}]{\Gamma\Ra\psi}{
	\Gamma\Ra\varphi
	&
	\varphi,\Gamma\Ra\psi
}
$$
\end{theorem}

\begin{lstlisting}[texcl]
Theorem G4iSLt_cut_adm : forall $\varphi$ $\Gamma$0 $\Gamma$1 $\chi$,
 (G4iSLt_prv ($\Gamma$0++$\Gamma$1,$\varphi$) * G4iSLt_prv ($\Gamma$0++$\varphi$::$\Gamma$1,$\chi$)) ->
               G4iSLt_prv ($\Gamma$0++$\Gamma$1,$\chi$).
\end{lstlisting}

\begin{proof}
Let~$d_{1}$ (with last rule~$r_{1}$) and~$d_{2}$ (with last rule~$r_{2}$) be \prfs~in~$\sys{G4iSLt}$ of  $\Gamma\Ra\varphi$ and $\varphi,\Gamma\Ra \chi$ respectively, as shown below.
\begin{center}
\AxiomC{$d_1$}
\RightLabel{$r_{1}$}
\UnaryInfC{$\Gamma \Ra\varphi$}
\DisplayProof
\qquad
\AxiomC{$d_2$}
\RightLabel{$r_{2}$}
\UnaryInfC{$\varphi, \Gamma\Ra \chi$}
\DisplayProof
\end{center}
We show that there is a \prf~in~$\sys{G4iSLt}$ of~$\Gamma\Ra \chi$.
We reason by strong primary induction (PI) on the weight of the cut-formula $\varphi$, giving the primary inductive hypothesis (PIH).
We also use a strong secondary induction (SI) on $\Theta(\Gamma\Ra \chi)$ of the conclusion of a cut, giving the secondary inductive hypothesis (SIH). 
Crucially, by using SIH we avoid the issues caused by the diagonal formula~\cite{Valentini83,GoreRamanayake12}.

We consider $r_1$. In total, there are thirteen cases for $r_1$: one for each rule in $\sys{G4iSLt}$. However, we can reduce the number of cases to eight. We separate them by using Roman numerals and showcase the most interesting ones. We show more cases in the  \hyperref[appendix]{Appendix}.

\noindent \textbf{(V)} $\bold{r_1=}\textbf{\rule{$\rightarrow$R}}:$ Then $r_1$ has the following form where $\varphi=\varphi_0\rightarrow\varphi_1$:
\begin{center}
\AxiomC{$\varphi_0,\Gamma\Ra\varphi_1$}
\RightLabel{$\scriptstyle\rule{$\rightarrow$R}$}
\UnaryInfC{$\Gamma\Ra\varphi_0\rightarrow\varphi_1$}
\DisplayProof
\end{center}
For the cases where $\varphi_0\rightarrow\varphi_1$ is principal in $r_2$ and $r_2\neq\rule{$\Box\!\rightarrow$L}$, or where $r_2\in\{\rule{IdP},\rule{$\bot$L}\}$, we refer to Dyckhoff and Negri's proof \cite{DycNeg00} as the cuts produced in these cases involve the traditional induction hypothesis PIH. We are left with seven sub-cases, but here again focus on the most interesting ones.

\noindent\textbf{(V-d)} If $r_2$ is \rule{$\rightarrow\rightarrow$L} where the cut formula is not principal in $r_2$, then it must have the following form where $(\gamma_0\rightarrow\gamma_1)\rightarrow\gamma_2,\Gamma_0=\Gamma$.
\begin{center}
\AxiomC{$\varphi_0\rightarrow\varphi_1,\gamma_1\rightarrow\gamma_2,\Gamma_0\Ra\gamma_0\rightarrow\gamma_1$}
\AxiomC{$\varphi_0\rightarrow\varphi_1,\gamma_2,\Gamma_0\Ra\chi$}
\RightLabel{$\scriptstyle\rule{$\rightarrow\rightarrow$L}$}
\BinaryInfC{$\varphi_0\rightarrow\varphi_1,(\gamma_0\rightarrow\gamma_1)\rightarrow\gamma_2,\Gamma_0\Ra\chi$}
\DisplayProof
\end{center}
Thus, $\Gamma\Ra\chi$ is of the form $(\gamma_0\rightarrow\gamma_1)\rightarrow\gamma_2,\Gamma_0\Ra\chi$ and $\Gamma\Ra\varphi_0\rightarrow\varphi_1$ is of the form $(\gamma_0\rightarrow\gamma_1)\rightarrow\gamma_2,\Gamma_0\Ra\varphi_0\rightarrow\varphi_1$. Using the admissible rule \rule{$\rightarrow\rightarrow$LIR} on the latter we obtain a proof of the sequent $\gamma_2,\Gamma_0\Ra\varphi_0\rightarrow\varphi_1$. Then consider the following proof of the sequent $\gamma_1\rightarrow\gamma_2,\Gamma_0\Ra \gamma_0\rightarrow\gamma_1$, where the rule \rule{$\rightarrow\rightarrow$LIL} deconstructs the implication $(\gamma_0\rightarrow\gamma_1)\rightarrow\gamma_2$, rule \rule{Ctr} contracts $\gamma_1\rightarrow\gamma_2$ and Lemma \ref{lem:inv_rules} is the invertibility of the rule \rule{$\rightarrow$R}.

\begin{scriptsize}
\begin{center}
\AxiomC{$(\gamma_0\rightarrow\gamma_1)\rightarrow\gamma_2,\Gamma_0\Ra\varphi_0\rightarrow\varphi_1$}
\dashedLine
\RightLabel{$\scriptstyle\rule{$\rightarrow\rightarrow$LIL}$}
\UnaryInfC{$\gamma_0,\gamma_1\rightarrow\gamma_2,\gamma_1\rightarrow\gamma_2,\Gamma_0\Ra\varphi_0\rightarrow\varphi_1$}
\dashedLine
\RightLabel{$\scriptstyle\rule{Ctr}$}
\UnaryInfC{$\gamma_0,\gamma_1\rightarrow\gamma_2,\Gamma_0\Ra\varphi_0\rightarrow\varphi_1$}
\AxiomC{$\varphi_0\rightarrow\varphi_1,\gamma_1\rightarrow\gamma_2,\Gamma_0\Ra\gamma_0\rightarrow\gamma_1$}
\dashedLine
\RightLabel{$\scriptstyle\text{Lem.\ref{lem:inv_rules}}$}
\UnaryInfC{$\varphi_0\rightarrow\varphi_1,\gamma_0,\gamma_1\rightarrow\gamma_2,\Gamma_0\Ra\gamma_1$}
\dashedLine
\RightLabel{$\scriptstyle\text{SIH}$}
\BinaryInfC{$\gamma_0,\gamma_1\rightarrow\gamma_2,\Gamma_0\Ra\gamma_1$}
\RightLabel{$\scriptstyle\rule{$\rightarrow$R}$}
\UnaryInfC{$\gamma_1\rightarrow\gamma_2,\Gamma_0\Ra\gamma_0\rightarrow\gamma_1$}
\DisplayProof
\end{center}
\end{scriptsize}
The crucial point here is to see that the use of SIH is justified, in other words, that $\Theta(\gamma_0,\gamma_1\rightarrow\gamma_2,\Gamma_0\Ra\gamma_1)\slex\Theta((\gamma_0\rightarrow\gamma_1)\rightarrow\gamma_2,\Gamma_0\Ra\chi)$. This is the case as the rule applications \rule{$\rightarrow\rightarrow$L} and \rule{$\rightarrow$R} entail $\Theta(\gamma_0,\gamma_1\rightarrow\gamma_2,\Gamma_0\Ra\gamma_1)$ $\slex\Theta(\gamma_1\rightarrow\gamma_2,\Gamma_0\Ra\gamma_0\rightarrow\gamma_1)\slex\Theta((\gamma_0\rightarrow\gamma_1)\rightarrow\gamma_2,\Gamma_0\Ra\chi)$ by Lemma \ref{lem:meas-decr-rule}, hence $\Theta(\gamma_0,\gamma_1\rightarrow\gamma_2,\Gamma_0\Ra\gamma_1)\slex\Theta((\gamma_0\rightarrow\gamma_1)\rightarrow\gamma_2,\Gamma_0\Ra\chi)$ by transitivity of~$\slex$. So, we are done. Note that the created cut could not be justified by usual induction on height, as the admissibility of \rule{$\rightarrow\rightarrow$LIL} is not height-preserving.

\noindent\textbf{(V-f)} If $r_2$ is \rule{$\Box\!\rightarrow$L} with a principal formula different from the cut formula, then it must have the following form where $\Box\gamma_0\rightarrow\gamma_1,\Phi,\Box\Gamma_0=\Gamma$.
\begin{center}
\AxiomC{$\varphi_0\rightarrow\varphi_1,\gamma_1,\Phi,\Gamma_0,\Box\gamma_0\Ra\gamma_0$}
\AxiomC{$\gamma_1,\varphi_0\rightarrow\varphi_1,\Phi,\Box\Gamma_0\Ra\chi$}
\RightLabel{$\scriptstyle\rule{$\Box\!\rightarrow$L}$}
\BinaryInfC{$\varphi_0\rightarrow\varphi_1,\Box\gamma_0\rightarrow\gamma_1,\Phi,\Box\Gamma_0\Ra\chi$}
\DisplayProof
\end{center}
Thus, we have that $\Gamma\Ra\chi$ and $\Gamma\Ra\varphi_0\rightarrow\varphi_1$ are respectively of the form $\Box\gamma_0\rightarrow\gamma_1,\Phi,\Box\Gamma_0\Ra\chi$ and $\Box\gamma_0\rightarrow\gamma_1,\Phi,\Box\Gamma_0\Ra\varphi_0\rightarrow\varphi_1$. Using the admissible rule \rule{$\Box\!\rightarrow$LIR} on the latter we obtain a proof of $\gamma_1,\Phi,\Box\Gamma_0\Ra\varphi_0\rightarrow\varphi_1$. Then, we proceed as follows by combining the proof $\pi$ second-below with the first one.
\vspace{-0.8cm}

\begin{center}
\AxiomC{$\pi$}
\noLine
\UnaryInfC{$\gamma_1,\Phi,\Gamma_0,\Box\gamma_0\Ra\gamma_0$}

\AxiomC{$\gamma_1,\Phi,\Box\Gamma_0\Ra\varphi_0\rightarrow\varphi_1$}
\AxiomC{$\gamma_1,\varphi_0\rightarrow\varphi_1,\Phi,\Box\Gamma_0\Ra\chi$}
\dashedLine
\RightLabel{$\scriptstyle\text{SIH}$}
\BinaryInfC{$\gamma_1,\Phi,\Box\Gamma_0\Ra\chi$}
\RightLabel{$\scriptstyle\rule{$\Box\!\rightarrow$L}$}
\BinaryInfC{$\Box\gamma_0\rightarrow\gamma_1,\Phi,\Box\Gamma_0\Ra\chi$}
\DisplayProof
\end{center}

\begin{center}
\AxiomC{$\varphi_0,\Box\gamma_0\rightarrow\gamma_1,\Phi,\Box\Gamma_0\Ra\varphi_1$}
\dashedLine
\RightLabel{$\scriptstyle\rule{Wkn}$}
\UnaryInfC{$\varphi_0,\Box\gamma_0\rightarrow\gamma_1,\Phi,\Box\Gamma_0,\Box\gamma_0\Ra\varphi_1$}
\dashedLine
\RightLabel{$\scriptstyle\rule{$\Box\!\rightarrow$LIR}$}
\UnaryInfC{$\varphi_0,\gamma_1,\Phi,\Box\Gamma_0,\Box\gamma_0\Ra\varphi_1$}
\dashedLine
\RightLabel{$\scriptstyle\rule{$\boxtimes$}$}
\UnaryInfC{$\varphi_0,\gamma_1,\Phi,\Gamma_0,\Box\gamma_0\Ra\varphi_1$}
\RightLabel{$\scriptstyle\rule{$\imp$R}$}
\UnaryInfC{$\gamma_1,\Phi,\Gamma_0,\Box\gamma_0\Ra\varphi_0\rightarrow\varphi_1$}
\AxiomC{$\varphi_0\rightarrow\varphi_1,\gamma_1,\Phi,\Gamma_0,\Box\gamma_0\Ra\gamma_0$}
\dashedLine
\RightLabel{$\scriptstyle\text{SIH}$}
\BinaryInfC{$\gamma_1,\Phi,\Gamma_0,\Box\gamma_0\Ra\gamma_0$}
\DisplayProof
\end{center}
Note that both uses of SIH are justified here, as the last rule in the first \prf~is an instance of \rule{$\Box\!\rightarrow$L} hence $\Theta(\gamma_1,\Phi,\Box\Gamma_0\Ra\chi)\slex\Theta(\Box\gamma_0\rightarrow\gamma_1,\Phi,\Box\Gamma_0\Ra\chi)$ and $\Theta(\gamma_1,\Phi,\Gamma_0,\Box\gamma_0\Ra\gamma_0)\slex\Theta(\Box\gamma_0\rightarrow\gamma_1,\Phi,\Box\Gamma_0\Ra\chi)$ by Lemma \ref{lem:meas-decr-rule}.

\noindent \textbf{(VII)} $\bold{r_1=}\textbf{\rule{$\Box\!\rightarrow$L}}$: Then $r_1$ is as follows, where $\Box\gamma_0\rightarrow\gamma_1,\Phi,\Box\Gamma_0=\Gamma$.
\begin{center}
\AxiomC{$\gamma_1,\Phi,\Gamma_0,\Box\gamma_0\Ra\gamma_0$}
\AxiomC{$\gamma_1,\Phi,\Box\Gamma_0\Ra\varphi$}
\RightLabel{$\scriptstyle\rule{$\Box\!\rightarrow$L}$}
\BinaryInfC{$\Box\gamma_0\rightarrow\gamma_1,\Phi,\Box\Gamma_0\Ra\varphi$}
\DisplayProof
\end{center}
Thus, the sequents $\Gamma\Ra\chi$ and $\varphi,\Gamma\Ra\chi$ are of the form $\Box\gamma_0\rightarrow\gamma_1,\Phi,\Box\Gamma_0\Ra\chi$ and $\varphi,\Box\gamma_0\rightarrow\gamma_1,\Phi,\Box\Gamma_0\Ra\chi$, respectively. Then, we proceed as follows.

\begin{center}
\AxiomC{$\gamma_1,\Phi,\Gamma_0,\Box\gamma_0\Ra\gamma_0$}
\AxiomC{$\gamma_1,\Phi,\Box\Gamma_0\Ra\varphi$}
\AxiomC{$\varphi,\Box\gamma_0\rightarrow\gamma_1,\Phi,\Box\Gamma_0\Ra\chi$}
\dashedLine
\RightLabel{$\scriptstyle\rule{$\Box\!\rightarrow$LIR}$}
\UnaryInfC{$\varphi,\gamma_1,\Phi,\Box\Gamma_0\Ra\chi$}
\dashedLine
\RightLabel{$\scriptstyle\text{SIH}$}
\BinaryInfC{$\gamma_1,\Phi,\Box\Gamma_0\Ra\chi$}
\RightLabel{$\scriptstyle\rule{$\Box\!\rightarrow$L}$}
\BinaryInfC{$\Box\gamma_0\rightarrow\gamma_1,\Phi,\Box\Gamma_0\Ra\chi$}
\DisplayProof
\end{center}
Note that the use of SIH is justified, as the last rule in this \prf~gives us $\Theta(\gamma_1,\Phi,\Box\Gamma_0\Ra\chi)\slex\Theta(\Box\gamma_0\rightarrow\gamma_1,\Phi,\Box\Gamma_0\Ra\chi)$ by Lemma \ref{lem:meas-decr-rule}.

\noindent \textbf{(VIII)} $\bold{r_1=}\textbf{\rule{SLtR}}$: Then $\varphi$ is the diagonal formula in $r_1$:
\begin{center}
\AxiomC{$\Phi,\Gamma_0,\Box\varphi_0\Ra\varphi_0$}
\RightLabel{$\scriptstyle\rule{SLtR}$}
\UnaryInfC{$\Phi,\Box\Gamma_0\Ra\Box\varphi_0$}
\DisplayProof
\end{center}
where $\varphi=\Box\varphi_0$ and $\Phi,\Box\Gamma_0=\Gamma$. Thus, we have that $\Gamma\Ra\chi$ and $\varphi,\Gamma\Ra\chi$ are respectively of the form $\Phi,\Box\Gamma_0\Ra\chi$ and  $\Box\varphi_0,\Phi,\Box\Gamma_0\Ra\chi$. We now consider $r_2$. 

\noindent\textbf{(VIII-b)} If $r_2$ is \rule{$\Box\!\rightarrow$L} it is of the following form, where $\Phi=\Box\gamma_0\rightarrow\gamma_1,\Phi_0$.
\begin{center}
\AxiomC{$\gamma_1,\Phi_0,\Box\gamma_0,\varphi_0,\Gamma_0\Ra\gamma_0$}
\AxiomC{$\gamma_1,\Phi_0,\Box\varphi_0,\Box\Gamma_0\Ra\chi$}
\RightLabel{$\scriptstyle\rule{$\Box\!\rightarrow$L}$}
\BinaryInfC{$\Box\gamma_0\rightarrow\gamma_1,\Phi_0,\Box\varphi_0,\Box\Gamma_0\Ra\chi$}
\DisplayProof
\end{center}
We proceed as follows.
\begin{scriptsize}
\begin{center}
\AxiomC{$\pi_0$}
\noLine
\UnaryInfC{$\gamma_1,\Phi_0,\Gamma_0,\Box\gamma_0\Ra\gamma_0$}

\AxiomC{$\Box\gamma_0\imp\gamma_1,\Phi_0,\Box\Gamma_0\Ra\Box\varphi_0$}
\dashedLine
\RightLabel{$\scriptstyle\rule{$\Box\!\rightarrow$LIR}$}
\UnaryInfC{$\gamma_1,\Phi_0,\Box\Gamma_0\Ra\Box\varphi_0$}
\AxiomC{$\gamma_1,\Phi_0,\Box\varphi_0,\Box\Gamma_0\Ra\chi$}
\dashedLine
\RightLabel{$\scriptstyle\text{SIH}$}
\BinaryInfC{$\gamma_1,\Phi_0,\Box\Gamma_0\Ra\chi$}

\RightLabel{$\scriptstyle\rule{$\Box\!\rightarrow$L}$}
\BinaryInfC{$\Box\gamma_0\rightarrow\gamma_1,\Phi_0,\Box\Gamma_0\Ra\chi$}
\DisplayProof
\end{center}
\end{scriptsize}
where $\pi_0$ is the first proof given below, which depends $\pi_1$, the second one:
\begin{center}
\AxiomC{$\Box\gamma_0\imp\gamma_1,\Phi_0,\Box\Gamma_0\Ra\Box\varphi_0$}
\dashedLine
\RightLabel{$\scriptstyle\rule{$\boxtimes$}$}
\UnaryInfC{$\Box\gamma_0\imp\gamma_1,\Phi_0,\Gamma_0\Ra\Box\varphi_0$}
\dashedLine
\RightLabel{$\scriptstyle\rule{Wkn}$}
\UnaryInfC{$\Box\gamma_0\imp\gamma_1,\Phi_0,\Gamma_0,\Box\gamma_0\Ra\Box\varphi_0$}
\dashedLine
\RightLabel{$\scriptstyle\rule{$\Box\!\rightarrow$LIR}$}
\UnaryInfC{$\gamma_1,\Phi_0,\Gamma_0,\Box\gamma_0\Ra\Box\varphi_0$}

\AxiomC{$\pi_1$}
\noLine
\UnaryInfC{$\gamma_1,\Phi_0,\Box\gamma_0,\Box\varphi_0,\Gamma_0\Ra\gamma_0$}

\dashedLine
\RightLabel{$\scriptstyle\text{SIH}$}
\BinaryInfC{$\gamma_1,\Phi_0,\Gamma_0,\Box\gamma_0\Ra\gamma_0$}
\DisplayProof
\end{center}
\begin{center}
\AxiomC{$\Box\gamma_0\imp\gamma_1,\Phi_0,\Box\varphi_0,\Gamma_0\Ra\varphi_0$}
\dashedLine
\RightLabel{$\scriptstyle\rule{Wkn}$}
\UnaryInfC{$\Box\gamma_0\imp\gamma_1,\Phi_0,\Box\gamma_0,\Box\varphi_0,\Gamma_0\Ra\varphi_0$}
\dashedLine
\RightLabel{$\scriptstyle\rule{$\Box\!\rightarrow$LIR}$}
\UnaryInfC{$\gamma_1,\Phi_0,\Box\gamma_0,\Box\varphi_0,\Gamma_0\Ra\varphi_0$}
\AxiomC{$\varphi_0,\gamma_1,\Phi_0,\Box\gamma_0,\Gamma_0\Ra\gamma_0$}
\dashedLine
\RightLabel{$\scriptstyle\rule{Wkn}$}
\UnaryInfC{$\varphi_0,\gamma_1,\Phi_0,\Box\gamma_0,\Box\varphi_0,\Gamma_0\Ra\gamma_0$}
\dashedLine
\RightLabel{$\scriptstyle\text{PIH}$}
\BinaryInfC{$\gamma_1,\Phi_0,\Box\gamma_0,\Box\varphi_0,\Gamma_0\Ra\gamma_0$}
\DisplayProof
\end{center}
Note that both uses of SIH are justified here as the rule application \rule{$\Box\!\rightarrow$L} entails $\Theta(\gamma_1,\Phi_0,\Gamma_0,\Box\gamma_0\Ra\gamma_0)\slex\Theta(\Box\gamma_0\rightarrow\gamma_1,\Phi_0,\Box\Gamma_0\Ra\chi)$ and we have $\Theta(\gamma_1,\Phi_0,\Box\Gamma_0\Ra\chi)\slex\Theta(\Box\gamma_0\rightarrow\gamma_1,\Phi_0,\Box\Gamma_0\Ra\chi)$ by Lemma \ref{lem:meas-decr-rule}.

\noindent\textbf{(VIII-c)} If $r_2$ is \rule{SLtR}, then it is of the following form where $\chi=\Box\chi_0$.
\begin{center}
\AxiomC{$\Phi,\varphi_0,\Gamma_0,\Box\chi_0\Ra\chi_0$}
\RightLabel{$\scriptstyle\rule{SLtR}$}
\UnaryInfC{$\Phi,\Box\varphi_0,\Box\Gamma_0\Ra\Box\chi_0$}
\DisplayProof
\end{center}
We proceed as follows.
\vspace{-0.4cm}

\begin{center}
\begin{scriptsize}
\AxiomC{$\Phi,\Gamma_0,\Box\varphi_0\Ra\varphi_0$}
\RightLabel{$\scriptstyle\rule{SLtR}$}
\UnaryInfC{$\Phi,\Box\Gamma_0\Ra\Box\varphi_0$}
\dashedLine
\RightLabel{$\scriptstyle\rule{$\boxtimes$}$}
\UnaryInfC{$\Phi,\Gamma_0\Ra\Box\varphi_0$}
\dashedLine
\RightLabel{$\scriptstyle\rule{Wkn}$}
\UnaryInfC{$\Phi,\Gamma_0,\Box\chi_0\Ra\Box\varphi_0$}

\AxiomC{$\Box\varphi_0,\Phi,\Gamma_0\Ra\varphi_0$}
\dashedLine
\RightLabel{$\scriptstyle\rule{Wkn}$}
\UnaryInfC{$\Box\varphi_0,\Phi,\Gamma_0,\Box\chi_0\Ra\varphi_0$}
\AxiomC{$\varphi_0,\Phi,\Gamma_0,\Box\chi_0\Ra\chi_0$}
\dashedLine
\RightLabel{$\scriptstyle\rule{Wkn}$}
\UnaryInfC{$\varphi_0,\Box\varphi_0,\Phi,\Gamma_0,\Box\chi_0\Ra\chi_0$}
\dashedLine
\RightLabel{$\scriptstyle\text{PIH}$}
\BinaryInfC{$\Box\varphi_0,\Phi,\Gamma_0,\Box\chi_0\Ra\chi_0$}

\dashedLine
\RightLabel{$\scriptstyle\text{SIH}$}
\BinaryInfC{$\Phi,\Gamma_0,\Box\chi_0\Ra\chi_0$}
\RightLabel{$\scriptstyle\rule{SLtR}$}
\UnaryInfC{$\Phi,\Box\Gamma_0\Ra\Box\chi_0$}
\DisplayProof
\end{scriptsize}
\end{center}
The use of SIH is justified because the last rule in this \prf~ensures that $\Theta(\Phi,\Gamma_0,\Box\chi_0\Ra\chi_0)\slex\Theta(\Phi,\Box\Gamma_0\Ra\Box\chi_0)$ by Lemma \ref{lem:meas-decr-rule}.
\qed
\end{proof}

The attentive reader may have noticed that our proof technique requires the use of additive, and not multiplicative, cuts. Indeed, the use of SIH relies on the decrease of the measure $\Theta$, which is notably ensured by the upward application of any rule of the calculus. 
More generally, in the proof of admissibility if the cut we initially consider has $\Gamma\Ra\chi$ as conclusion, then we can justify a cut with conclusion $\Gamma'\Ra\chi'$ using SIH as long as we have a chain $r_0,\dots,r_n$ of application of rules of $\sys{G4iSLt}$ of the following form. 
$$
\infer[\scriptstyle r_0]{\Gamma\Ra\chi}{
	\ldots
	&
	\infer[\scriptstyle r_n]{\vdots}{
		\ldots
		&
		\Gamma'\Ra\chi'
		&
		\ldots
	}
	&
	\ldots
}
$$
However, the contraction rule does not ensure the decrease of the measure $\Theta$ from conclusion to premise: it is not the case that $\Theta(\Gamma,\varphi,\varphi\Ra\chi)\slex\Theta(\Gamma,\varphi\Ra\chi)$. So, this prevents us from allowing one of $r_0,\dots,r_n$ above to be $\rule{Ctr}$. This is where multiplicative cuts are problematic: they most often use the contraction rule as follows, where $\Gamma\Ra\chi$ is the conclusion of the initial cut and $\Gamma',\Gamma''\Ra\chi'$ is the conclusion of the cut we want to justify through SIH. 
\begin{center}
\AxiomC{$\Gamma'\Ra\varphi$}
\AxiomC{$\varphi,\Gamma''\Ra\chi'$}
\dashedLine
\RightLabel{$\scriptstyle\text{SIH}$}
\BinaryInfC{$\Gamma',\Gamma''\Ra\chi'$}
\noLine
\UnaryInfC{$\vdots$}
\dashedLine
\RightLabel{$\scriptstyle\rule{Ctr}^*$}
\UnaryInfC{$\Gamma\Ra\chi$}
\DisplayProof
\end{center}
Unfortunately, the presence of the contraction rule above $\Gamma\Ra\chi$ disallows us from using SIH on $\Gamma',\Gamma''\Ra\chi'$, as we are not ensured that the measure decreased between the two sequents. So, our proof technique prohibited us from using multiplicative cuts, forcing us to use additive ones. This observation was already made by Gor\'{e} and Shillito \cite{GorShi22}.

Using our purely syntactic proof of cut-admissibility above, we easily obtain a cut-elimination procedure for the calculus $\sys{G4iSLt}$ extended with \rule{cut}, by simply repetitively eliminating topmost cuts first. To effectively prove this statement in Coq we explicitly encode the additive cut rule as follows:
\begin{center}
\AxiomC{\lstinline!($\Gamma$0++$\Gamma$1 *\ $\varphi$)!}
\AxiomC{\lstinline!($\Gamma$0++$\varphi$::$\Gamma$1 *\ $\chi$)!}
\BinaryInfC{\lstinline!($\Gamma$0++$\Gamma$1 *\ $\chi$)!}
\DisplayProof
\end{center}

We encode the calculus $\sys{G4iSLt}+\rule{cut}$ as \lstinline!G4iSLt_cut_rules!, i.e.~\lstinline!G4iSLt_rules! enhanced with \rule{cut}.
Finally, we turn to the elimination of additive cuts:

\begin{theorem}
The additive cut rule is eliminable from $\sys{G4iSLt}+\rule{cut}$.
\end{theorem}
\vspace{-0.2cm}
\begin{lstlisting}[texcl]
Theorem G4iSLt_cut_elimination : forall s,
 (G4iSLt_cut_prv s) -> (G4iSLt_prv s).
\end{lstlisting}

The above theorem shows that any proof in $\sys{G4iSLt}+\rule{cut}$ of a sequent, i.e.~\lstinline!G4iSLt_cut_prv s!, can be transformed into a proof in $\sys{G4iSLt}$ of the same sequent.
As this theorem is in fact a constructive function based on \lstinline{Type}, we can use the extraction feature of Coq and obtain a cut-eliminating Haskell program.

\section{Conclusion}

This paper introduces a sequent calculus for $\sys{iSL}$, denoted $\sys{G4iSLt}$. It is an improvement over the sequent calculus $\sys{G4iSL}$ from~\cite{GieIem23}, because backward proof search for $\sys{G4iSLt}$ is strongly terminating (instead of weakly terminating) shown via a new well-founded measure, and cut-elimination is proved directly (instead of indirectly via an equivalent calculus based on $\sys{G3i}$~\cite{GieIem23}). All our results are formalised in Coq in a constructive way. In turn, Coq's extraction mechanism can generate a Haskell program for the cut-elimination procedure for $\sys{G4iSLt}$.

One of the reasons to develop $\sys{G4iSLt}$ is to use its strongly terminating proof search to investigate uniform interpolation, a strengthening of Craig interpolation, in the setting of intuitionistic provability logics.
Typically, calculi with good (weakly or strongly) terminating proof search form good grounds for constructive proofs of uniform interpolation (see e.g.~\cite{Pitts92,Bil06,Iemhoff17,AkbJal18,AkbIemJal21b,GieJalKuz21,AfsGraMen21,AkbIemJal22}).

We also suggest to develop a countermodel construction for $\sys{G4iSLt}$ similarly to the one for $\sys{G4iSL}$ in \cite{GieIem23}. Furthermore, as $\sys{iSL}$ is an intuitionistic modal logic only defined with $\Box$, there is the question how it can be extended by $\Diamond$ operators. It is clear from the literature of intuitionistic modal logics that several choices can be made (e.g. \cite{FischerServi77,Simpson94PhD,Bellin2001,Mendler2005,WolterZakharyaschev97}), so we leave this for future work.

\subsubsection{Acknowledgements}
Iris van der Giessen would like to thank Sonia Marin and Marianna Girlando for an interesting discussion on the subtle choice of rules in proof systems. 
We would like to thank the anonymous reviewers for their helpful comments and suggestions.
Van der Giessen is supported by a UKRI Future Leaders Fellowship, ‘Structure vs Invariants in Proofs’, project reference MR/S035540/1. 
Rosalie Iemhoff is supported  by the Netherlands Organisation for Scientific Research under grant 639.073.807 and by the EU H2020-MSCA-RISE-2020 Project 101007627. 
Rajeev Gor\'e is supported by FWF project P 33548 and the National Centre for Research and Development, Poland (NCBR), and the Luxembourg National Research Fund (FNR), under the PolLux/FNR-CORE project STV (POLLUX-VII/1/2019).

\bibliographystyle{splncs04}
\providecommand{\noopsort}[1]{}

\newpage
\appendix

\section*{Appendix} \label{appendix}

\begin{proof}[of Lemma \ref{lem:meas-decr-rule}]
We showcase the difficult case: the rule \rule{SLtR}. Our goal is to show the following.
$$\Theta(\Phi,\Gamma,\Box\varphi\Ra\varphi)\;\;\slex\;\;\Theta(\Phi,\Box\Gamma\Ra\Box\varphi)$$
We first unbox the antecedent of each sequent before generating the list of natural numbers, thus obtaining the multisets $\Phi\uplus\Gamma\uplus\{\varphi,\varphi\}$ and $\Phi\uplus\Gamma\uplus\{\Box\varphi\}$. One can see that the occurrences of formulae in $\Phi$ and $\Gamma$ are not going to be decisive in comparing the corresponding lists of natural numbers, as they appear in both multisets. So, we can truncate our multisets to $\{\varphi,\varphi\}$ and $\{\Box\varphi\}$. But we have already seen above that the lists generated by these two multisets are such that the list generated from $\{\varphi,\varphi\}$ is smaller on $\slex$ than the one of $\{\Box\varphi\}$, as the single occurrence of the formula $\Box\varphi$ of weight $n$ is replaced by two occurences of the formula $\varphi$ of lesser weight $n-1$.
\qed
\end{proof}

\begin{proof}[of Theorem \ref{thm-CE-G4iSLt}]
\noindent \textbf{(I)} $\bold{r_1=}\textbf{\rule{IdP}}:$ then we have that $\varphi=p$. Consequently, $\Gamma\Ra\chi$ is of the form $\Gamma_0,p\Ra\chi$. Also, the conclusion of $r_2$ is of the form $\Gamma_0,p,p\Ra\chi$. We can apply the contraction rule \rule{Ctr} to obtain a \prf~of $\Gamma_0,p\Ra\chi$. 

\noindent \textbf{(II)} $\bold{r_1=}\textbf{\rule{$\bot$L}}$: Then $r_1$ must have the following form where $\Gamma_0,\bot=\Gamma$:
\begin{center}
\AxiomC{}
\RightLabel{$\scriptstyle\rule{$\bot$L}$}
\UnaryInfC{$\Gamma_0,\bot\Ra\varphi$}
\DisplayProof
\end{center}
Thus, we have that the sequent $\Gamma\Ra\chi$ is of the form $\Gamma_0,\bot\Ra\chi$, and is an instance of \rule{$\bot$L}. So we are done.

\noindent \textbf{(III)} $\bold{r_1\in}\{\textbf{\rule{$\land$L}, \rule{$\lor$L}, \rule{$p\!\rightarrow$L}, \rule{$\land\!\rightarrow$L}, \rule{$\lor\!\rightarrow$L}}\}$: In all these cases, the cut formula is not principal in $r_1$ so it is preserved in the premise. Given that the rules considered are invertible, we simply take the conclusion of $r_2$ and use the corresponding invertibility lemma to destruct the principal formula of $r_1$. Then, we use SIH to cut on $\varphi$ in the obtained premises, and apply $r_1$ on the conclusion of the cut. As an example, we consider the case of \rule{$\land$L}, where $r_1$ is of the following form and where $\Gamma_0,\psi\land\delta=\Gamma$:
\begin{center}
\AxiomC{$\Gamma_0,\psi,\delta\Ra\varphi$}
\RightLabel{$\scriptstyle\rule{$\land$L}$}
\UnaryInfC{$\Gamma_0,\psi\land\delta\Ra\varphi$}
\DisplayProof
\end{center}
Thus, we have that the sequents $\Gamma\Ra\chi$ and $\varphi,\Gamma\Ra\chi$ are respectively of the form $\Gamma_0,\psi\land\delta\Ra\chi$ and $\varphi,\Gamma_0,\psi\land\delta\Ra\chi$. Using the invertibility of \rule{$\land$L}, proven in Lemma \ref{lem:inv_rules}, on $\varphi,\Gamma_0,\psi\land\delta\Ra\chi$ we obtain a proof of the sequent $\varphi,\Gamma_0,\psi,\delta\Ra\chi$. Then, we proceed as follows.
\begin{center}
\AxiomC{$\Gamma_0,\psi,\delta\Ra\varphi$}
\AxiomC{$\varphi,\Gamma_0,\psi,\delta\Ra\chi$}
\dashedLine
\RightLabel{$\scriptstyle\text{SIH}$}
\BinaryInfC{$\Gamma_0,\psi,\delta\Ra\chi$}
\RightLabel{$\scriptstyle\rule{$\land$L}$}
\UnaryInfC{$\Gamma_0,\psi\land\delta\Ra\chi$}
\DisplayProof
\end{center}
Note that the use of SIH is justified here as the last rule in this \prf~entails
$\Theta(\Gamma_0,\psi,\delta\Ra\chi)\slex\Theta(\Gamma_0,\psi\land\delta\Ra\chi)$ by Lemma \ref{lem:meas-decr-rule}.

\noindent \textbf{(IV)} $\bold{r_1\in}\{\textbf{\rule{$\land$R}, \rule{$\lor$R$_1$}, \rule{$\lor$R$_2$}}\}$: In all these cases, the cut formula is principal in $r_1$ so it is deconstructed in the premise. Given that the corresponding left rules are invertible, we simply take the conclusion of $r_2$ and use the adequate invertibility lemma to destruct the cut formula. Then, we use PIH to cut on the obtained subformulae. As an example, we consider the case of \rule{$\lor$R$_1$}, where $r_1$ is of the following form and where $\psi\lor\delta=\varphi$:
\begin{center}
\AxiomC{$\Gamma\Ra\psi$}
\RightLabel{$\scriptstyle\rule{$\lor$R$_1$}$}
\UnaryInfC{$\Gamma\Ra\psi\lor\delta$}
\DisplayProof
\end{center}
Thus, we have that the sequents $\varphi,\Gamma\Ra\chi$ is of the form $\psi\lor\delta,\Gamma\Ra\chi$. Using the invertibility of \rule{$\lor$L}, proven in Lemma \ref{lem:inv_rules}, on $\psi\lor\delta,\Gamma\Ra\chi$ we obtain proofs of the sequents $\psi,\Gamma\Ra\chi$ and $\delta,\Gamma\Ra\chi$. Then, we proceed as follows.
\begin{center}
\AxiomC{$\Gamma\Ra\psi$}
\AxiomC{$\psi,\Gamma\Ra\chi$}
\dashedLine
\RightLabel{$\scriptstyle\text{PIH}$}
\BinaryInfC{$\Gamma\Ra\chi$}
\DisplayProof
\end{center}

\noindent\textbf{(V-a)} If $r_2$ is \rule{$\rightarrow$R} then it must have the following form where $\chi_0\rightarrow\chi_1=\chi$:
\begin{center}
\AxiomC{$\varphi_0\rightarrow\varphi_1,\chi_0,\Gamma\Ra\chi_1$}
\RightLabel{$\scriptstyle\rule{$\rightarrow$R}$}
\UnaryInfC{$\varphi_0\rightarrow\varphi_1,\Gamma\Ra\chi_0\rightarrow\chi_1$}
\DisplayProof
\end{center}
We can use the weakening rule \rule{Wkn} on the
\prf~of $\Gamma\Ra\varphi_0\rightarrow\varphi_1$
to get a \prf~of $\chi_0,\Gamma\Ra\varphi_0\rightarrow\varphi_1$
of no greater height. Proceed as follows.

\begin{center}
\AxiomC{$\chi_0,\Gamma\Ra\varphi_0\rightarrow\varphi_1$}
\AxiomC{$\varphi_0\rightarrow\varphi_1,\chi_0,\Gamma\Ra\chi_1$}
\dashedLine
\RightLabel{$\scriptstyle\text{SIH}$}
\BinaryInfC{$\chi_0,\Gamma\Ra\chi_1$}
\RightLabel{$\scriptstyle\rule{$\rightarrow$R}$}
\UnaryInfC{$\Gamma\Ra\chi_0\rightarrow\chi_1$}
\DisplayProof
\end{center}
Note the use of the additive cut to contract on $\chi_0$.
The use of SIH is justified here as the last rule in this \prf~gives us $\Theta(\chi_0,\Gamma\Ra\chi_1)\slex\Theta(\Gamma\Ra\chi_0\rightarrow\chi_1)$ by Lemma \ref{lem:meas-decr-rule}.

\noindent\textbf{(V-b)} If $r_2$ is \rule{$\land$R} or \rule{$\lor$R$_i$}, then we simply use cut with the premise(s) of $r_2$ and the conclusion of $r_1$ using SIH. As an example, we consider the case of \rule{$\lor$R$_1$}, where $r_2$ has the form
where $\chi = \chi_0\lor\chi_1$:
\begin{center}
\AxiomC{$\varphi_0\rightarrow\varphi_1,\Gamma\Ra\chi_0$}
\RightLabel{$\scriptstyle\rule{$\lor$R$_1$}$}
\UnaryInfC{$\varphi_0\rightarrow\varphi_1,\Gamma\Ra\chi_0\lor\chi_1$}
\DisplayProof
\end{center}
Then we proceed as follows:
\begin{center}
\AxiomC{$\Gamma\Ra\varphi_0\rightarrow\varphi_1$}
\AxiomC{$\varphi_0\rightarrow\varphi_1,\Gamma\Ra\chi_0$}
\dashedLine
\RightLabel{$\scriptstyle\text{SIH}$}
\BinaryInfC{$\Gamma\Ra\chi_0$}
\RightLabel{$\scriptstyle\rule{$\lor$R$_1$}$}
\UnaryInfC{$\Gamma\Ra\chi_0\lor\chi_1$}
\DisplayProof
\end{center}

\noindent\textbf{(V-c)} If $r_2$ is \rule{$\land$L}, \rule{$\lor$L}, \rule{$p\!\rightarrow$L}, \rule{$\lor\!\rightarrow$L} or \rule{$\land\!\rightarrow$L} where the cut formula is not principal in $r_2$, then we use the inversion lemma for $r_2$ on the conclusion of $r_1$, and then apply cut using SIH. As an example, we consider the case of \rule{$\land\!\rightarrow$L}, where $r_2$ has the form
where $(\gamma_0\land\gamma_1)\rightarrow\gamma_2,\Gamma_0=\Gamma$:
\begin{center}
\AxiomC{$\varphi_0\rightarrow\varphi_1,\gamma_0\rightarrow(\gamma_1\rightarrow\gamma_2),\Gamma_0\Ra\chi$}
\RightLabel{$\scriptstyle\rule{$\land\!\rightarrow$L}$}
\UnaryInfC{$\varphi_0\rightarrow\varphi_1,(\gamma_0\land\gamma_1)\rightarrow\gamma_2,\Gamma_0\Ra\chi$}
\DisplayProof
\end{center}
Thus, we have that the sequents $\Gamma\Ra\chi$ and $\Gamma\Ra\varphi_0\rightarrow\varphi_1$ are respectively of the form $(\gamma_0\land\gamma_1)\rightarrow\gamma_2,\Gamma_0\Ra\chi$ and  $(\gamma_0\land\gamma_1)\rightarrow\gamma_2,\Gamma_0\Ra\varphi_0\rightarrow\varphi_1$.
Using the invertibility of \rule{$\land\!\rightarrow$L}, proven in Lemma \ref{lem:inv_rules}, on $(\gamma_0\land\gamma_1)\rightarrow\gamma_2,\Gamma_0\Ra\varphi_0\rightarrow\varphi_1$
we obtain a proof of the sequent $\gamma_0\rightarrow(\gamma_1\rightarrow\gamma_2),\Gamma_0\Ra\varphi_0\rightarrow\varphi_1$. Then, we proceed as follows.
\begin{center}
\AxiomC{$\gamma_0\rightarrow(\gamma_1\rightarrow\gamma_2),\Gamma_0\Ra\varphi_0\rightarrow\varphi_1$}
\AxiomC{$\varphi_0\rightarrow\varphi_1,\gamma_0\rightarrow(\gamma_1\rightarrow\gamma_2),\Gamma_0\Ra\chi$}
\dashedLine
\RightLabel{$\scriptstyle\text{SIH}$}
\BinaryInfC{$\gamma_0\rightarrow(\gamma_1\rightarrow\gamma_2),\Gamma_0\Ra\chi$}
\RightLabel{$\scriptstyle\rule{$\land\!\rightarrow$L}$}
\UnaryInfC{$(\gamma_0\land\gamma_1)\rightarrow\gamma_2,\Gamma_0\Ra\chi$}
\DisplayProof
\end{center}

\noindent\textbf{(V-e)} If $r_2$ is \rule{$\Box\!\rightarrow$L} with the cut formula as principal formula, then it must have the following form, where $\Phi,\Box\Gamma_0=\Gamma$ and $\varphi_0=\Box\varphi_2$.
\begin{center}
\AxiomC{$\varphi_1,\Phi,\Gamma_0,\Box\varphi_2\Ra\varphi_2$}
\AxiomC{$\varphi_1,\Phi,\Box\Gamma_0\Ra\chi$}
\RightLabel{$\scriptstyle\rule{$\Box\!\rightarrow$L}$}
\BinaryInfC{$\Box\varphi_2\rightarrow\varphi_1,\Phi,\Box\Gamma_0\Ra\chi$}
\DisplayProof
\end{center}
Thus, we have that the sequents $\Gamma\Ra\chi$ and $\varphi_0,\Gamma\Ra\varphi_1$ are respectively of the form $\Phi,\Box\Gamma_0\Ra\chi$ and $\Box\varphi_2,\Phi,\Box\Gamma_0\Ra\varphi_1$. Then, we proceed as follows by combining the proof $\pi$ second-below with the first one.
\begin{center}
\AxiomC{$\pi$}
\noLine
\UnaryInfC{$\Phi,\Box\Gamma_0\Ra\varphi_1$}
\AxiomC{$\varphi_1,\Phi,\Box\Gamma_0\Ra\chi$}
\dashedLine
\RightLabel{$\scriptstyle\text{PIH}$}
\BinaryInfC{$\Phi,\Box\Gamma_0\Ra\chi$}
\DisplayProof
\end{center}

\begin{center}
\AxiomC{$\Phi,\Box\Gamma_0,\Box\varphi_2\Ra\varphi_1$}
\dashedLine
\RightLabel{$\scriptstyle\rule{$\boxtimes$}$}
\UnaryInfC{$\Phi,\Gamma_0,\Box\varphi_2\Ra\varphi_1$}
\AxiomC{$\varphi_1,\Phi,\Gamma_0,\Box\varphi_2\Ra\varphi_2$}
\dashedLine
\RightLabel{$\scriptstyle\text{PIH}$}
\BinaryInfC{$\Phi,\Gamma_0,\Box\varphi_2\Ra\varphi_2$}
\RightLabel{$\scriptstyle\rule{SLtR}$}
\UnaryInfC{$\Phi,\Box\Gamma_0\Ra\Box\varphi_2$}
\AxiomC{$\Box\varphi_2,\Phi,\Box\Gamma_0\Ra\varphi_1$}
\dashedLine
\RightLabel{$\scriptstyle\text{PIH}$}
\BinaryInfC{$\Phi,\Box\Gamma_0\Ra\varphi_1$}
\DisplayProof
\end{center}

\noindent\textbf{(V-g)} If $r_2$ is \rule{SLtR} then it must have the following form.
\begin{center}
\AxiomC{$\Phi,\varphi_0\rightarrow\varphi_1,\Gamma_0,\Box\chi_0\Ra\chi_0$}
\RightLabel{$\scriptstyle\rule{SLtR}$}
\UnaryInfC{$\Phi,\varphi_0\rightarrow\varphi_1,\Box\Gamma_0\Ra\Box\chi_0$}
\DisplayProof
\end{center}
where $\Phi,\Box\Gamma_0=\Gamma$ and $\Box\chi_0=\chi$. In that case, note that the sequents $\Gamma\Ra\chi$ and $\Gamma\Ra\varphi_0\rightarrow\varphi_1$ are of the respective form $\Phi,\Box\Gamma_0\Ra\Box\chi_0$ and $\Phi,\Box\Gamma_0\Ra\varphi_0\rightarrow\varphi_1$. Then, we proceed as follows.
\begin{center}
\AxiomC{$\Phi,\Box\Gamma_0\Ra\varphi_0\rightarrow\varphi_1$}
\dashedLine
\RightLabel{$\scriptstyle\rule{$\boxtimes$}$}
\UnaryInfC{$\Phi,\Gamma_0\Ra\varphi_0\rightarrow\varphi_1$}
\dashedLine
\RightLabel{$\scriptstyle\rule{Wkn}$}
\UnaryInfC{$\Phi,\Gamma_0,\Box\chi_0\Ra\varphi_0\rightarrow\varphi_1$}
\AxiomC{$\varphi_0\rightarrow\varphi_1,\Phi,\Gamma_0,\Box\chi_0\Ra\chi_0$}
\dashedLine
\RightLabel{$\scriptstyle\text{SIH}$}
\BinaryInfC{$\Phi,\Gamma_0,\Box\chi_0\Ra\chi_0$}
\RightLabel{$\scriptstyle\rule{SLtR}$}
\UnaryInfC{$\Phi,\Box\Gamma_0\Ra\Box\chi_0$}
\DisplayProof
\end{center}
Note that the use of SIH is justified here as the last rule in this \prf~gives us $\Theta(\Phi,\Gamma_0,\Box\chi_0\Ra\chi_0)\slex\Theta(\Phi,\Box\Gamma_0\Ra\Box\chi_0)$ by Lemma \ref{lem:meas-decr-rule}.

\noindent \textbf{(VI)} $\bold{r_1=}\textbf{\rule{$\rightarrow\rightarrow$L}}$: Then $r_1$ is as follows, where $(\gamma_0\rightarrow\gamma_1)\rightarrow\gamma_2,\Gamma_0=\Gamma$.
\begin{center}
\AxiomC{$\gamma_1\rightarrow\gamma_2,\Gamma_0\Ra\gamma_0\rightarrow\gamma_1$}
\AxiomC{$\gamma_2,\Gamma_0\Ra\varphi$}
\RightLabel{$\scriptstyle\rule{$\rightarrow\rightarrow$L}$}
\BinaryInfC{$(\gamma_0\rightarrow\gamma_1)\rightarrow\gamma_2,\Gamma_0\Ra\varphi$}
\DisplayProof
\end{center}
Thus, we have that the sequents $\Gamma\Ra\chi$ and $\varphi,\Gamma\Ra\chi$ are respectively of the form $(\gamma_0\rightarrow\gamma_1)\rightarrow\gamma_2,\Gamma_0\Ra\chi$ and  $\varphi,(\gamma_0\rightarrow\gamma_1)\rightarrow\gamma_2,\Gamma_0\Ra\chi$. Using the admissible rule \rule{$\rightarrow\rightarrow$LIR} on $\varphi,(\gamma_0\rightarrow\gamma_1)\rightarrow\gamma_2,\Gamma_0\Ra\chi$ we obtain a proof of the sequent $\varphi,\gamma_2,\Gamma_0\Ra\chi$. Then, we proceed as follows.
\begin{center}
\AxiomC{$\gamma_1\rightarrow\gamma_2,\Gamma_0\Ra\gamma_0\rightarrow\gamma_1$}
\AxiomC{$\gamma_2,\Gamma_0\Ra\varphi$}
\AxiomC{$\varphi,\gamma_2,\Gamma_0\Ra\chi$}
\dashedLine
\RightLabel{$\scriptstyle\text{SIH}$}
\BinaryInfC{$\gamma_2,\Gamma_0\Ra\chi$}
\RightLabel{$\scriptstyle\rule{$\rightarrow\rightarrow$L}$}
\BinaryInfC{$(\gamma_0\rightarrow\gamma_1)\rightarrow\gamma_2,\Gamma_0\Ra\chi$}
\DisplayProof
\end{center}
Note that the use of SIH is justified as the last rule in this \prf~gives us $\Theta(\gamma_2,\Gamma_0\Ra\chi)\slex\Theta((\gamma_0\rightarrow\gamma_1)\rightarrow\gamma_2,\Gamma_0\Ra\chi)$ by Lemma \ref{lem:meas-decr-rule}.

\noindent\textbf{(VIII-a)} If $r_2$ is one of \rule{IdP}, \rule{$\bot$L}, \rule{$\land$R}, \rule{$\land$L}, \rule{$\lor$R$_1$}, \rule{$\lor$R$_2$}, \rule{$\lor$L}, \rule{$\rightarrow$R}, \rule{$p\!\rightarrow$L}, \rule{$\land\!\rightarrow$L}, \rule{$\lor\!\rightarrow$L} and \rule{$\rightarrow\rightarrow$L} then proceed similarly to the cases (I), (II), (III), (IV) and (VI), where the cut-formula is not principal in the rules considered by using SIH.
\qed
\end{proof}

\end{document}